\PassOptionsToPackage{unicode}{hyperref}
\PassOptionsToPackage{hyphens}{url}
\PassOptionsToPackage{dvipsnames,svgnames,x11names}{xcolor}
\documentclass[
  11pt]{article}

\usepackage{amsmath,amssymb}
\usepackage{amsthm, mathtools}
\usepackage{iftex}
\ifPDFTeX
  \usepackage[T1]{fontenc}
  \usepackage[utf8]{inputenc}
  \usepackage{textcomp} % provide euro and other symbols
\else % if luatex or xetex
  \usepackage{unicode-math}
  \defaultfontfeatures{Scale=MatchLowercase}
  \defaultfontfeatures[\rmfamily]{Ligatures=TeX,Scale=1}
\fi
\usepackage{lmodern}
\ifPDFTeX\else  
\fi
\IfFileExists{upquote.sty}{\usepackage{upquote}}{}
\IfFileExists{microtype.sty}{% use microtype if available
  \usepackage[]{microtype}
  \UseMicrotypeSet[protrusion]{basicmath} % disable protrusion for tt fonts
}{}
\usepackage{xcolor}
\makeatletter
\ifx\paragraph\undefined\else
  \let\oldparagraph\paragraph
  \renewcommand{\paragraph}{
    \@ifstar
      \xxxParagraphStar
      \xxxParagraphNoStar
  }
  \newcommand{\xxxParagraphStar}[1]{\oldparagraph*{#1}\mbox{}}
  \newcommand{\xxxParagraphNoStar}[1]{\oldparagraph{#1}\mbox{}}
\fi
\ifx\subparagraph\undefined\else
  \let\oldsubparagraph\subparagraph
  \renewcommand{\subparagraph}{
    \@ifstar
      \xxxSubParagraphStar
      \xxxSubParagraphNoStar
  }
  \newcommand{\xxxSubParagraphStar}[1]{\oldsubparagraph*{#1}\mbox{}}
  \newcommand{\xxxSubParagraphNoStar}[1]{\oldsubparagraph{#1}\mbox{}}
\fi
\makeatother

\usepackage{longtable,booktabs,array}
\usepackage{calc} % for calculating minipage widths
\usepackage{etoolbox}
\makeatletter
\patchcmd\longtable{\par}{\if@noskipsec\mbox{}\fi\par}{}{}
\makeatother
\IfFileExists{footnotehyper.sty}{\usepackage{footnotehyper}}{\usepackage{footnote}}
\makesavenoteenv{longtable}
\usepackage{graphicx}
\makeatletter
\def\maxwidth{\ifdim\Gin@nat@width>\linewidth\linewidth\else\Gin@nat@width\fi}
\def\maxheight{\ifdim\Gin@nat@height>\textheight\textheight\else\Gin@nat@height\fi}
\makeatother
\setkeys{Gin}{width=\maxwidth,height=\maxheight,keepaspectratio}
\makeatletter
\def\fps@figure{htbp}
\makeatother

\usepackage{fullpage}
\makeatletter
\@ifpackageloaded{caption}{}{\usepackage{caption}}
\AtBeginDocument{%
\ifdefined\contentsname
  \renewcommand*\contentsname{Table of contents}
\else
  \newcommand\contentsname{Table of contents}
\fi
\ifdefined\listfigurename
  \renewcommand*\listfigurename{List of Figures}
\else
  \newcommand\listfigurename{List of Figures}
\fi
\ifdefined\listtablename
  \renewcommand*\listtablename{List of Tables}
\else
  \newcommand\listtablename{List of Tables}
\fi
\ifdefined\figurename
  \renewcommand*\figurename{Figure}
\else
  \newcommand\figurename{Figure}
\fi
\ifdefined\tablename
  \renewcommand*\tablename{Table}
\else
  \newcommand\tablename{Table}
\fi
}
\@ifpackageloaded{float}{}{\usepackage{float}}
\floatstyle{ruled}
\@ifundefined{c@chapter}{\newfloat{codelisting}{h}{lop}}{\newfloat{codelisting}{h}{lop}[chapter]}
\floatname{codelisting}{Listing}

\makeatother
\makeatletter
\@ifpackageloaded{caption}{}{\usepackage{caption}}
\@ifpackageloaded{subcaption}{}{\usepackage{subcaption}}
\makeatother

\ifLuaTeX
  \usepackage{selnolig}  % disable illegal ligatures
\fi
\usepackage[round]{natbib}
\usepackage{bibunits}
\defaultbibliographystyle{apalike} % agsm
\defaultbibliography{ref,CKDR}

\usepackage{bookmark}

\IfFileExists{xurl.sty}{\usepackage{xurl}}{} % add URL line breaks if available
\hypersetup{
  pdftitle={Title},
  pdfauthor={Author 1; Author 2},
  pdfkeywords={3 to 6 keywords, that do not appear in the title},
  colorlinks=true,
  linkcolor={blue},
  filecolor={Maroon},
  citecolor={Blue},
  urlcolor={Blue},
  pdfcreator={LaTeX via pandoc}}

\usepackage{mathrsfs} % fancy maths latters

\usepackage{tikz-cd}
\tikzcdset{every label/.append style = {font = \scriptsize}} % font size on the arrow

\usepackage{tikz}
\usepackage{xstring}
\usetikzlibrary{fit, calc, positioning}

\usepackage{graphicx}
\usepackage{caption}
\usepackage{subcaption}
\usepackage{float}

\usepackage{booktabs} % for professional tables
\usepackage{tabularx} % for paragraph-like tables (long texts)
\usepackage{multirow}
\usepackage{array}

\usepackage[capitalize,noabbrev]{cleveref}

\usepackage{algorithm,algorithmic}

\theoremstyle{plain}
\newtheorem{theorem}{Theorem}
\newtheorem{proposition}{Proposition}
\newtheorem{lemma}{Lemma}
\newtheorem{corollary}{Corollary}
\theoremstyle{definition}
\newtheorem{definition}{Definition}
\newtheorem{remark}{Remark}

\newtheorem{assumption}{Assumption}

\usepackage{bbm}
\newcommand{\1}{\mathbf{1}}
\newcommand{\cx}{\mathcal{X}}
\newcommand{\cy}{\mathcal{Y}}
\newcommand{\cz}{\mathcal{Z}}
\newcommand{\ch}{\mathcal{H}}
\newcommand{\cg}{\mathcal{G}}
\newcommand{\cs}{\mathcal{S}}

\newcommand{\ce}{\mathcal{E}}

\newcommand{\cm}{\mathcal{M}}

\newcommand{\R}{\mathbb{R}}

\def\P{\mathbb{P}} % ignore existing commands
\def\E{\mathbb{E}}
\def\Q{\mathbb{Q}}
\def\s{\Sigma}
\newcommand{\ind}{\perp\!\!\!\!\perp}
\newcommand{\hatsig}{\widehat{\Sigma}}

\newcommand{\mbf}[1]{\mathbf{#1}}

\DeclareMathOperator{\var}{Var}
\DeclareMathOperator{\Tr}{Tr}
\DeclareMathOperator{\cov}{Cov}
\DeclareMathOperator{\ran}{ran}
\DeclareMathOperator{\row}{row}
\DeclareMathOperator{\supp}{supp}
\DeclareMathOperator{\sign}{sign}

\DeclareMathOperator{\rk}{rank} % rank (rk or rank; change depending on needs)
\DeclareMathOperator{\Gr}{Gr} % Grassmannian
\DeclareMathOperator{\St}{St} % Stiefel manifold

\DeclareMathOperator*{\argmin}{argmin} % * makes supscripts go to the bottom

\DeclareMathOperator*{\minimize}{minimize}

\DeclarePairedDelimiter{\abs}{|}{|}
\DeclarePairedDelimiter{\pr}{(}{)}
\DeclarePairedDelimiter{\br}{\{}{\}}
\DeclarePairedDelimiter{\bracket}{[}{]}
\DeclarePairedDelimiter{\norm}{\|}{\|}
\newcommand{\clo}[1]{\overline{#1}}
\newcommand{\wt}[1]{\widetilde{#1}}
\newcommand{\wh}[1]{\widehat{#1}}

\begin{document}

\def\spacingset#1{\renewcommand{\baselinestretch}%
{#1}\small\normalsize} 

\spacingset{1}

\title{Geometry-preserving and interpretable dimension reduction of compositional data}
\author{Junyoung Park$^1$, Cheolwoo Park$^{2*}$, and Jeongyoun Ahn$^3$\thanks{Corresponding authors: Cheolwoo Park (parkcw2021@kaist.ac.kr), Jeongyoun Ahn (jyahn@kaist.ac.kr)}
}
\date{\vspace{-.3em}\normalsize$^1$Department of Biostatistics, University of Michigan, $^2$Department of Mathematical Sciences, KAIST, $^3$Department of Industrial and Systems Engineering, KAIST
\vspace{-.3em}}
\maketitle

\begin{abstract}

High-dimensional compositional data pose unique statistical challenges due to the simplex constraint and excess zeros. While dimension reduction is indispensable for analyzing such data, conventional approaches often rely on log-ratio transformations that compromise interpretability and distort the data through ad hoc zero replacements. To address these issues, we introduce a geometry-preserving framework for dimension reduction of compositional data, mapping high-dimensional compositions directly to a lower-dimensional simplex. This framework is interpretable as a softened amalgamation of compositions and enables dual visualization---showing both projected data and how variables contribute to reduced components---for at-a-glance interpretation. Within this geometry, we define a new sufficient dimension reduction (SDR) approach for compositional predictors, whose identifiable object, termed the central compositional subspace, differs from the classical central subspace in Euclidean SDR. 
For estimation, we propose a kernel-based method that yields sparse solutions and comes with an intrinsic predictive model for direct downstream analyses. We prove consistency through a new subspace-comparison argument that allows the estimated and target subspaces to have different dimensions.
Applications to real microbiome datasets demonstrate that our approach provides a powerful graphical exploration tool for uncovering meaningful biological patterns in high-dimensional compositional data.
\end{abstract}

\noindent%

\spacingset{1.5}

\begin{bibunit}

\section{Introduction}\label{sec: intro}

Compositional data comprise relative proportions of variables and lie in the unit simplex $\Delta^{d-1} = \{(x_1,\ldots,x_{d})^\top\in\R^d \,|\, \sum_{i=1}^d x_i = 1,\ x_i\geq 0,\ \forall i \}$. Such data arise naturally in diverse scientific fields, including chemometrics \citep{greenacre2020amalgamations}, physical activity \citep{janssenSystematicReviewCompositional2020}, and text mining \citep{wuSparseTopicModeling2023}.
Human microbiome compositions, in particular, have attracted significant interest for their intricate associations with health conditions and diseases, including obesity, diabetes, and cancer \citep{huttenhower2012structure, petersonAnalysisMicrobiomeData2024}. 
They are typically obtained through high-throughput sequencing (e.g., 16S rRNA gene sequencing), which generates microbial taxon counts that are subsequently normalized to compositions to account for differences in total counts across samples.
However, their data-driven analysis is challenging due to the large number of microbial taxa---often exceeding the sample size---and the underlying simplex constraint. 
These challenges make dimension reduction essential, but also require that the reduced representations remain interpretable in terms of the original compositional variables.

Traditionally, dimension reduction for compositional data has largely relied on transformations that map the simplex to other spaces. The most widely used are log-ratio transformations \citep{aitchison1986statistical}, which convert compositions into log-ratios in Euclidean space and thereby enable standard methods such as principal component analysis (PCA) and sufficient dimension reduction \citep{tomassiSufficientDimensionReduction2021}. 
However, these transformations become invalid when zeros are prevalent, as is common in microbiome data \citep{lutz2022survey}. Although a common remedy replaces zeros with small positive values \citep{martin2011dealing}, such replacements are systematically amplified by subsequent log transformations, inducing data distortions and making analyses highly sensitive to the choice of zero-replacement scheme \citep{park2022kernel}. Furthermore, transformations compromise direct interpretability in terms of original compositional variables. For instance, principal components from log-ratio transformations take the form $\sum_{j=1}^d\beta_j\log x_j$ with $\sum_{j=1}^d\beta_j = 0$, but the log transformation and the interdependence among variables $x_j$ obscure how changes in the original composition affect the component \citep{liItsAllRelative2023}.
The same interpretational difficulty persists in other transformations, including spherical \citep{scealy2011regression, park2022kernel} and power transformations \citep{scealyRobustPrincipalComponent2015}, underscoring the trade-off between analytic convenience and interpretability.

To circumvent these issues, several transformation-free dimension reduction methods have been proposed recently.
\citet{kim2024principal} proposed a direct PCA approach that estimates nested subspaces intersecting the simplex and approximating compositions, and \citet{leePrincipalSubsimplexAnalysis2025} proposed a vertex aggregation method that seeks the best approximating subsimplex for given data. However, both methods rely on greedy stepwise algorithms that may yield suboptimal solutions.
Another important line of work employs amalgamation, which reduces dimension by merging variables into larger compositional categories, thereby preserving the simplex structure and offering direct interpretations \citep{aitchison1986statistical}. 
Although amalgamation has been largely overlooked because of its incompatibility with log-ratio-based linear models \citep{greenacre2020amalgamations}, recent work has developed data-driven procedures, including loss-based optimization with genetic algorithms \citep{quinnAmalgamsDatadrivenAmalgamation2020} and linear regression with parameter-fusion regularization \citep{liItsAllRelative2023}. 
However, amalgamation is inherently a discrete operation, which limits its representational flexibility and creates combinatorial optimization challenges. Consequently, the compositional data literature still lacks a dimension reduction approach that is simultaneously transformation-free, interpretable, and flexible enough to offer sufficient representational power.

In this work, we propose a new simplex-geometry-preserving and interpretable dimension reduction framework for compositional data. 
We begin by reinterpreting amalgamation as a discrete mass assignment from each vertex of the original simplex to a single vertex of a lower-dimensional simplex. We then relax this assignment to allow soft mass allocation across multiple vertices, yielding a flexible simplex-to-simplex reduction that we term \textit{compositional dimension reduction} (CDR). CDR is represented by column-stochastic matrices
\begin{equation}\label{eq: cdr matrix}
    \cm_{m, d} = \big\{P = (p_{ij})\in\R^{m\times d}\, \big|\, 0\leq p_{ij}\leq 1,\ \sum_{i=1}^m p_{ij} = 1, \ j=1,\ldots,d\big\},
\end{equation}
which map compositions in $\Delta^{d-1}$ directly to $\Delta^{m-1}$. This relaxation preserves the advantages of amalgamation: CDR avoids the zero-handling issues of log-ratio transformations, offers direct variable-level interpretability, and makes transparent how the reduction reflects changes in the original composition (Section~\ref{sec:compo-dimred}). Furthermore, it enables a \textit{dual visualization}: when plotting the reduced data on a ternary plot ($m=3$) using a CDR matrix $P\in\cm_{3,d}$, the columns of $P$ can also be plotted on a second ternary plot, providing immediate visual insight into how the original variables affect the reduced simplicial coordinates (Figure~\ref{fig: comparison with tomassi}). 
To our knowledge, this is the first work to use column-stochastic matrices for interpretable dimension reduction of compositional data; only \citet{wangDirichletComponentAnalysis2008} used such matrices for feature extraction, but without interpreting reductions directly using original variables.

\begin{figure}[t]
    \centering
    \includegraphics[width=\textwidth]{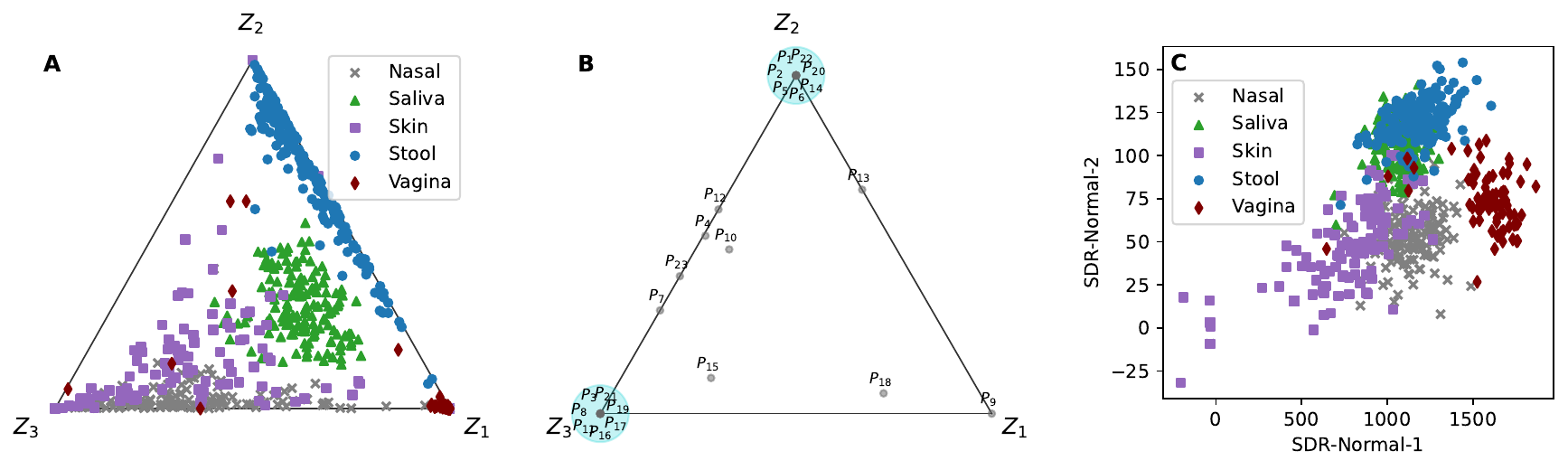}
    \vspace{-1.3em}
    \spacingset{1}
    \caption{
    Visualization of CDR versus the SDR-Normal method \citep{tomassiSufficientDimensionReduction2021} on the Human Microbiome Project data (Section~\ref{sec:visualization}). (A) Reduced data: samples cluster by body site on the simplex, in which the axes retain compositional meaning. (B) Visualization of the CDR matrix $P \in \cm_{3, 23}$ estimated from the proposed method (\cref{sec:compositional-kdr}): points represent columns $P_j$ of $P$, showing how original variables contribute to the three reduced axes. Cyan bubbles highlight clusters of similar columns, revealing amalgamation patterns. (C) Euclidean visualization from SDR-Normal: compositional interpretation is obscured, and axes are not visually interpretable.
    }
    \label{fig: comparison with tomassi}
\end{figure}

The softened CDR framework, however, introduces two major technical challenges for method development. 
First, although CDR is interpretable at the matrix level, the reduction matrix itself is not identifiable: for a compositional vector $X\in\Delta^{d-1}$, a matrix $P\in\cm_{m,d}$, and an invertible matrix $Q\in\cm_{m,m}$, the two CDR matrices $P$ and $QP \in\cm_{m,d}$ induce equivalent reductions, since $PX$ and $QPX$ are in one-to-one correspondence.
Second, because the domain $\cm_{m,d}$ contains matrices of variable rank, learning a CDR matrix $P\in\cm_{m,d}$ must avoid rank collapse below the data's intrinsic dimension.
In Euclidean dimension reduction, these issues are typically resolved using orthogonality constraints or subspace-based arguments. However, these tools do not transfer to CDR: orthogonality is incompatible with the column-stochastic constraint, and no analogous subspace theory exists for reductions of this form.
Similar difficulties have long been recognized as challenging in nonnegative matrix factorization, where identifiability requires additional geometric or sparsity conditions, and controlling rank collapse often relies on stringent assumptions that exclude degenerate or rank-deficient regimes \citep{fuIdentifiabilityNonnegativeMatrix2018, barbarinoRobustnessMinimumVolumeNonnegative2025}.

In this work, we address these challenges in the supervised setting with a response $Y$: we resolve the identifiability issue by developing a simplex-tailored sufficient dimension reduction (SDR), and handle the rank collapse via a novel asymptotic theory. Specifically, we define \textit{compositional SDR} for $X\in\Delta^{d-1}$ by the constrained conditional independence:
\begin{equation}\label{eq:compo-sdr} 
    Y\ind X \,|\, PX, \quad P\in\cm_{m,d}.
\end{equation}
This states that $PX$ retains all information in $X$ relevant for predicting $Y$ \citep{li1991sliced}, while the reduction matrix $P$ is constrained to be column-stochastic. The constraint $P\in\cm_{m,d}$ is essential and makes compositional SDR fundamentally depart from unconstrained Euclidean SDR: without it, the central subspace---the primary target of Euclidean SDR \citep{cook1998regression}---fails to exist for compositional predictors (see Section~\ref{sec:compositional-SDR}). In contrast, we prove that the constraint $P\in\cm_{m,d}$ recasts the classical SDR foundations within the native simplex geometry, defining a new identifiable objective called the \textit{central compositional subspace}.

To estimate this, we extend the principles of kernel dimension reduction (KDR) \citep{fukumizuKernelDimensionReduction2009}, which minimizes a reproducing kernel Hilbert space (RKHS)-based conditional dependence measure. Classical KDR asymptotic theory relies on re-embedding reduced data into the original Euclidean domain using semiorthogonal matrices, a construction with no direct analogue under simplex geometry. We thus reformulate KDR directly on the target simplex of reduction by deriving an equivalent form of the conditional dependence measure. The resulting method, compositional KDR (CKDR), optimizes this measure over the matrix domain $\cm_{m,d}$. 
Unlike prior KDR methods with semiorthogonal matrices, the polytopic geometry of $\cm_{m,d}$ further induces emergent sparsity in the optimized matrix, analogously to optimization over $\ell^1$-balls \citep{tibshirani1996regression}; these sparse patterns reveal underlying amalgamation structures and make the fitted matrix easier to interpret.

On this foundation, we make three additional contributions that establish CKDR as a practical and theoretically grounded method.
First, {we reveal an intrinsic predictive model underlying the CKDR objective, linked to vector-valued kernel ridge regression \citep{micchelli2005learning}. This model enables direct downstream prediction with strong empirical performance}, whereas earlier KDR methods were coupled with separate prediction procedures \citep{chen2017kernel, park2023kernel}. Second, we use this intrinsic predictive model to solve the nonconvex CKDR optimization via successive convex approximation (SCA) \citep{scutariParallelDistributedMethods2017}; in their original form, KDR objectives do not readily admit such approximations since the Gram matrix inversion involved obstructs typical kernel-linearization techniques \citep{allen2013automatic}.
Third, we close the remaining varying-rank challenge by developing a novel asymptotic theory for CKDR. Using a distance between subspaces of different dimensions \citep{yeSchubertVarietiesDistances2016}, we prove that the estimator's row space eventually contains the central compositional subspace and does not collapse below the target dimension.

Although compositional SDR offers a general formulation, obtaining an estimator by extending Euclidean SDR methods is difficult. Classical SDR methods often directly estimate the central subspace, which is undefined for compositional predictors. Optimization-based SDR approaches \citep{dongBriefReviewLinear2021} avoid this issue by targeting conditional independence directly, but their fixed-rank asymptotic analyses do not directly accommodate the varying-rank CDR domain $\cm_{m,d}$. 
Imposing a fixed rank on $\cm_{m,d}$ may seem simpler, but it destroys the compact convex structure of the full domain, complicating both estimation and theory. Our CKDR provides a principled solution to these obstacles, combining optimization over the full CDR domain with theory that controls its varying-rank behavior.

We previously developed a KDR-based method for variable selection of compositional data in \citet{park2023kernel}, which amalgamated unselected variables to preserve simplex geometry. Its discrete formulation, however, made the theory substantially simpler by avoiding the identifiability and varying-rank challenges inherent to the softened CDR framework. It also lacked the intrinsic predictive model and associated optimization strategy introduced in this paper. Nonetheless, its continuous relaxation-based algorithm substantially improved upon log-ratio-based variable selection methods, which motivated us to develop the present general, flexible yet interpretable framework with rigorous theoretical support.

{In summary, our primary novel contributions are: (i) a geometry-preserving, flexible, and interpretable dimension reduction framework for compositional data; (ii) a compositional SDR formulation with a new identifiable target; (iii) a CKDR estimator equipped with an intrinsic predictive model for downstream prediction and optimization; and (iv) a subspace-comparison argument establishing consistency under varying rank in the CDR domain $\cm_{m,d}$.}
Empirically, CKDR performs strongly on real microbiome data and yields interpretable low-dimensional representations that reveal biologically meaningful patterns (\cref{sec: experiments}). Together, this geometry-respecting dimension reduction approach paves the way for new interpretable insights into high-dimensional compositional data.

\section{Interpretability of the CDR framework}\label{sec:compo-dimred}

This section details the interpretability of the compositional dimension reduction (CDR) framework. \cref{sec:interpretability} elaborates its interpretation in connection with amalgamation, effectively illustrated through dual visualization in \cref{sec:visualization}.

\subsection{Interpretation as soft amalgamation}\label{sec:interpretability}

Amalgamation reduces the dimensionality of a composition $x = (x_1,\ldots,x_d)^\top\in\Delta^{d-1}$ by aggregating its components into $m \le d$ mutually exclusive and collectively exhaustive groups \citep{aitchison1986statistical}. This operation is expressed as 
$$Ax = \bigg({\sum_{j: A_j= e_1} x_j,\ldots, \sum_{j: A_j = e_m} x_j}\bigg)^\top\in\Delta^{m-1},$$
where $A = [A_1,\ldots, A_d]\in\cm_{m,d}$ is a binary, column-stochastic matrix. Each column $A_j$ is a standard basis vector $e_i\in\R^m$, which implements a \textit{hard assignment} of the component $x_j$ entirely to the $i$th aggregated variable (amalgam) $z_i$ in the reduced composition $z = Ax$ (solid arrows in \cref{fig: neural net}).

\begin{figure}
    \centering
    \footnotesize
    \begin{subfigure}{0.4\linewidth}
    \begin{tikzpicture}[
    shorten >=1pt, shorten <=1pt,
    legend/.style={font=\large\bfseries},
    scale = 1
  ]
    \foreach \i in {1,...,6}
    {
        \node[circle, draw, thick, fill=gray!12, minimum size=7mm] (in-\i) at (0,-\i*1) {$x_{\i}$};
    }
    \foreach \i in {1,...,4}
        \node[circle, draw, thick, fill=gray!12, minimum size=7mm] (out-\i) at (2.1,-\i*1 -1) {$z_{\i}$};

    \draw[->] (in-1) -- (out-1) node[midway, above=0.1] (x1) {\footnotesize $\!\!\!\!\!\quad +x_1{e}_1$};
    \draw[->] (in-2) -- (out-1);
    \draw[->] (in-4) -- (out-1);
    \draw[->] (in-3) -- (out-3);
    \draw[->] (in-5) -- (out-4) node[midway, above=-0.5mm] (x1) {\footnotesize $+x_5{e}_4$};
    \draw[->] (in-6) -- (out-4) node[midway, below=0.05] (x1) {\footnotesize $+x_6{e}_4$};

    \node (domain) at (0.07,0.1) {\large $\Delta^5$};
    \node (image) at (2.2, 0.1) {\large $\Delta^3$};
    \draw[->] ([yshift=-1.5mm] domain) -- (image) node[midway, above=0] (A) {$A$};

    \node[right=0 of out-1] {$=x_1+x_2+x_4$};
    \node[right=0 of out-2] {$=0$};
    \node[right=0 of out-3] {$=x_3$};
    \node[right=0 of out-4] {$=x_5+x_6$};

    \end{tikzpicture}
    \end{subfigure}
    $\qquad $\begin{subfigure}{0.2\linewidth}
    \begin{tikzpicture}[
        shorten >=1pt, shorten <=1pt,
        legend/.style={font=\large\bfseries},
        scale = 1
      ]
    \foreach \i in {1,...,6}
    {
        \node[circle, draw, thick, fill=gray!15, minimum size=7mm] (in-\i) at (0,-\i*1) {$x_{\i}$};
    }
    \foreach \i in {1,...,4}
        \node[circle, draw, thick, fill=gray!15, minimum size=7mm] (out-\i) at (2.1,-\i*1 -1) {$z_{\i}$};

    \foreach \i in {1,6}{
        \foreach \j in {1,...,4}{
            \draw[->, dashed, darkgray] (in-\i) -- (out-\j); % gray!40 etc...
        }
    }
    \node (dots) at (0.8, -3.4) {\large $\vdots$};
    \draw[dashed, darkgray] (in-1) -- (out-1) node[midway, above=0.05] (x1) {\footnotesize\textcolor{black}{$\!\!\!\quad+x_1P_1$}};
    \draw[dashed, darkgray] (in-6) -- (out-4) node[midway, below=0.05] (x6) {\footnotesize\textcolor{black}{$\!\!\!\quad+x_6P_6$}};

    \node (domain) at (0.07, 0.1) {\large $\Delta^5$};
    \node (image) at (2.2, 0.1) {\large $\Delta^3$};
    \draw[->] ([yshift=-1.5mm] domain) -- (image) node[midway, above=0] (A) {$P$};

    \end{tikzpicture}
    \end{subfigure}
    \spacingset{1}
    \caption{Illustration of compositional dimension reduction from $x\in\Delta^5$ to $z\in\Delta^3$. 
    \textit{Left}: Amalgamation $z = (x_1+x_2+x_4, 0, x_3, x_5+x_6)$  induced by a binary matrix $A = [e_1, e_1, e_3, e_1, e_4, e_4]\in\cm_{4,6}$ assigns each $x_j$ to a single component $z_i$, as depicted by the solid arrows.
    \textit{Right}: Soft amalgamation via a CDR matrix $P\in\cm_{4, 6}$ allocates each $x_j$ to multiple components according to the weights in  $P_j$, as depicted by the dashed arrows.
    }
    \label{fig: neural net}
\end{figure}
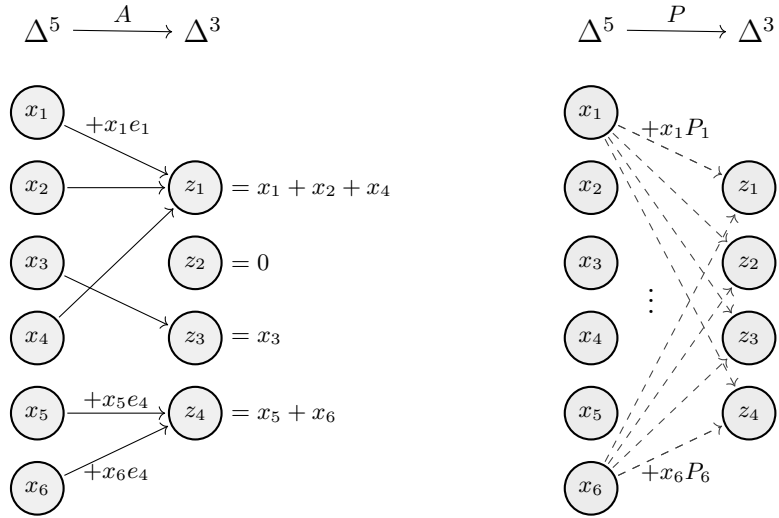

We view CDR as a relaxation of amalgamation that permits soft allocations. For a general column-stochastic matrix $P=[P_1,\ldots,P_d]\in\cm_{m,d}$, the CDR of a composition $x$ is $z:= Px \in\Delta^{m-1}$, for which each column $P_j\in \Delta^{m-1}$ acts as a weight vector, distributing the mass of $x_j$ across the $m$ components of the reduced composition (dashed arrows in \cref{fig: neural net}). Each entry $z_i$ of $z$ becomes a weighted sum of the original variables---a \textit{soft amalgam}---retaining the interpretability of amalgamation while providing greater flexibility.

The linear structure of CDR also enables a relative interpretation of effect sizes for changes in $x$, similar to recent work on transformation-free linear regression with compositions \citep{liItsAllRelative2023, fikselTransformationfreeLinearRegression2022}. Because variables of $x$ cannot vary independently, any change from $x$ to $x'$ is a zero-sum vector $\alpha = x' - x \in\R^d$, meaning any increases in one component must be offset by decreases in others. The induced change in $z$ is given by the linear contrast $P\alpha = \alpha_1 P_1 + \cdots + \alpha_d P_d$. For instance, if $x_j$ increases by $\delta$ while $x_k$ decreases by $\delta$, then $P\alpha = \delta (P_j - P_k)$, indicating that $z$ moves along the direction $P_j - P_k$ within the target simplex $\Delta^{m-1}$. If $P_j = P_k$, any such change preserving $x_j + x_k$ leaves the reduction $z$ unchanged, implying that $z$ views $x_j$ and $x_k$ as amalgamated into a single component. Thus, under this relative perspective, CDR encodes amalgamation not only through binary columns but also through similarity of columns. Because columns may occupy any point of the simplex---unlike hard amalgamation, which restricts them to vertices---the representation becomes substantially more flexible; even very low-dimensional CDRs, like projections to $\Delta^1$, can still capture arbitrary amalgamation patterns.
These simplex-native interpretations are difficult to recover in transformation-based approaches, because the zero-sum vector $\alpha = x' -x$ does not translate transparently into the transformed coordinates.

Sparsity further sharpens the link between CDR and amalgamation: (i) \emph{individual sparsity}, where each column $P_j$ contains many zeros, meaning $x_j$ influences only a few components of $z$; and (ii) \emph{equi-sparsity}, where identical columns reveal latent amalgamation structures. The geometry of the simplex, notably, connects these two types of sparsity, as strong individual sparsity tends to cluster columns near simplex vertices (Figure~\ref{fig: comparison with tomassi}B).

\subsection{Dual visualization of CDR}\label{sec:visualization}

The interpretability of CDR can be made visible through \textit{dual visualization} when the target dimension is $m=3$: the reduced compositions $z = Px$ are visualized on one ternary plot, while the columns $P_j \in \Delta^2$ are displayed on a second ternary plot, which we refer to as the variable allocation plot. As an example, we consider the Human Microbiome Project dataset from \citet{tomassiSufficientDimensionReduction2021}, consisting of $n = 681$ samples and $d = 23$ taxa collected from five body sites/specimen types: \texttt{nasal}, \texttt{saliva}, \texttt{skin}, \texttt{stool}, and \texttt{vagina}.

Figures \ref{fig: comparison with tomassi}A-B show the CDR result using a matrix $P \in \mathcal{M}_{3, 23}$ obtained by our method (\cref{sec:compositional-kdr}). The left panel displays the projected data, where the five classes are moderately separated. For example, high values of $z_1$, $z_2$, and $z_3$ are associated with \texttt{vagina}, \texttt{stool}, and \texttt{skin}/\texttt{nasal} samples, respectively. The middle panel shows the variable allocation plot, where most columns $P_j$ lie along the boundary of the simplex, indicating that each variable $x_j$ contributes primarily to one or two components of $z$. Notably, $P_9$ lies near the vertex $z_1$, identifying $x_9$ as particularly abundant in \texttt{vagina}. Similarly, seven columns cluster near $z_2$ (linked to \texttt{stool}) and seven near $z_3$ (linked to \texttt{skin}/\texttt{nasal}). The five columns near the middle of the left edge---$P_{4}, P_{7}, P_{10}, P_{12},$ and $P_{23}$---reflect an even distribution between $z_2$ and $z_3$, implying low abundance in \texttt{vagina} but limited utility in discriminating between the other classes.

\section{Compositional Sufficient Dimension Reduction}\label{sec:compositional-SDR}

While the CDR framework yields flexible and interpretable simplex-to-simplex reductions, its principled use must address the associated challenges: a CDR matrix is not identifiable, and the domain $\cm_{m,d}$ permits rank collapse. This section resolves the identifiability issue by extending sufficient dimension reduction (SDR) to the simplex geometry; Section~\ref{sec:consistency} addresses rank control through asymptotic theory for the estimator introduced in Section~\ref{sec:compositional-kdr}.

Traditionally, SDR is defined via the conditional independence relation $Y \ind X\,|\, BX,$ where $Y$ is a response, $X\in \R^d$ is a predictor, and $B\in\R^{m\times d}$ with $m\le d$ is a reduction matrix, indicating that $BX$ captures all information in $X$ that is relevant to predicting $Y$. As $Y\ind X\,|\, (QB)X$ holds for any invertible matrix $Q\in\mathbb{R}^{m\times m}$, the row space $\row(B)$ defines an equivalence class known as an SDR subspace. If the intersection of all SDR subspaces remains an SDR subspace, it is called the \emph{central subspace} \citep{cook1998regression}, which provides an identifiable target and has been the main focus of SDR approaches.

However, when $X$ is compositional, this structure breaks down: the central subspace does not exist
since the intersection of SDR subspaces always collapses to zero, making most existing SDR methods inapplicable. The following lemma formalizes this observation:
\begin{lemma}\label{lemma:zero-central-subs}
    For any response $Y$ and compositional predictor $X\in\Delta^{d-1}$, the intersection of SDR subspaces is always the zero subspace, thus does not form an SDR subspace.
\end{lemma}

This degeneracy stems from the affine structure of the simplex $\Delta^{d-1}$. Since $\1_d^\top X = 1$, where $\1_d = (1,\ldots,1)^\top$, deleting any coordinate $X_j$ loses no information about $X$; hence the corresponding coordinate-deletion matrix $B_{-j}$ is an SDR matrix, 
while $\cap_{j=1}^d\row(B_{-j})=0$ (see Appendix~\ref{sec:supp-sdr1} for full proof).
We show that this problem can be overcome by restricting the dimension reduction matrices to CDR matrices, whose row spaces necessarily contain $\1_d$. This leads to the constrained SDR framework defined below.

\begin{definition}[Compositional SDR]\label{def:compo-sdr}
    Let $X = (X_1,\ldots,X_d)^\top\in\Delta^{d-1}$ be a random compositional vector, and let $Y$ be a random response variable defined in a domain $\cy$. If 
    \begin{equation*}
        Y \ind X \,|\,PX, \qquad P\in\cm_{m,d},
    \end{equation*}
    where $m\le d$, we call $PX$ a \textit{compositional SDR} (CSDR) and $P$ a CSDR matrix.
\end{definition}

A weaker sufficiency at the conditional mean level can also be defined \citep{cook2002dimension}: if $\E[Y|X] = \E[Y|PX]$, $P\in\cm_{m,d}$, we call $PX$ a \emph{CSDR for conditional mean}. 
As in the traditional setting, multiple matrices can define the same CSDR, since any invertible matrix $Q\in\cm_{m,m}$ satisfying $QP\in\cm_{m,d}$ yields an equivalent reduction. For identifiability of CSDR, we focus on the row space $\row(P)$.
A subspace of $\mathbb{R}^d$ is called a CSDR subspace if it is the row space of a CSDR matrix. The following proposition shows that, under mild conditions on the distribution of $X$ on the simplex $\Delta^{d-1}$, intersections of CSDR subspaces remain CSDR subspaces, thus avoiding the degeneracy issue described in Lemma~\ref{lemma:zero-central-subs}.

\begin{proposition}\label{prop:intersection-csdr}
    Suppose that $X$ admits a density on $\Delta^{d-1}$ supported on a convex set with nonempty interior in $\Delta^{d-1}$. Then, the intersection of any collection of CSDR subspaces is itself a CSDR subspace.
\end{proposition}

In Appendix~\ref{sec:supp-central-compo-subsp}, we prove this result under even milder conditions. The proof largely parallels classical SDR arguments, with the main additional challenge being to show the existence of a CDR matrix whose row space equals the intersection of CSDR subspaces. The following lemma provides this existence, a key fact that underlies compositional SDR and our estimator's asymptotic guarantees (\cref{sec:consistency}).

\begin{lemma}\label{lemma:1subspace}
    Let $V$ be a subspace of $\R^d$ with $\dim V= m$ and $\1_d\in V$. Then, there exists a CDR matrix $P\in\cm_{m,d}$ with $\row(P)=V$.
\end{lemma}

For the remainder of the paper, we assume that the conclusion of Proposition~\ref{prop:intersection-csdr} holds. This ensures that the following compositional analogue of the minimal equivalence class of CSDR matrices is well-defined:
\begin{definition}\label{def:central-subsp}
    The \textit{central compositional subspace} is defined as the intersection of all CSDR subspaces and is denoted by $\mathcal{C}_{Y|X}$. 
\end{definition}

An analogous construction, obtained by replacing CSDR with CSDR for conditional mean, yields the {central mean compositional subspace} $\mathcal{C}_{\E[Y|X]}$, which is a subspace of $\mathcal{C}_{Y|X}$. As both subspaces contain $\1_d$, their dimensions are at least two unless $X\ind Y$.

We illustrate using the relative-shift linear model \citep{liItsAllRelative2023}: $Y = \sum_{j=1}^d \beta_jX_j +\varepsilon,$  where $X\in \Delta^{d-1}$, $Y\in\R$, and $\varepsilon \ind X$. Assume $\beta_1\leq\cdots\leq \beta_d $ without loss of generality and $\beta_1\neq \beta_d$ to avoid the independence between $Y$ and $X$. Define, for each $j=1, \ldots, d$,
\begin{equation}\label{eq:rs-es}
    P_j = \dfrac{\beta_d-\beta_j}{\beta_d-\beta_1} (1, 0)^\top + \dfrac{\beta_j-\beta_1}{\beta_d-\beta_1} (0, 1)^\top 
\end{equation}
and let $P = (P_1,\ldots, P_d)\in\cm_{2,d}$. Then, $Y = (\beta_1-\beta_d, 0)PX + \beta_d + \varepsilon$, establishing the SDR relations $Y\ind X \,|\, PX$ and $\E[Y|X] = \E[Y|PX]$. Since $X$ and $Y$ are not independent and $\E[Y|X]$ is non-constant, the dimension of $\row(P)$ is minimal among the CSDR subspaces. Therefore, we conclude that $\row(P) = \mathcal{C}_{Y|X} = \mathcal{C}_{\E[Y|X]}. $

\begin{remark}\label{rmk:suff-amalgam}
    Under sparsity, restricting CSDR matrices to binary matrices yields the analogous notions of \textit{sufficient amalgamation} and \textit{central amalgamation subspace}. In Appendix~\ref{sec:supp-suff-amalg},
    we prove the equivalence between sufficient amalgamation and column-wise equi-sparsity in a CSDR matrix, aligning with the discussion in Section~\ref{sec:interpretability}. This link offers a natural interpretation: clusters of nearly identical columns in the estimated CSDR matrix indicate an underlying sufficient amalgamation, so variables within each cluster can be merged without loss of predictive information (bubbles in Figure~\ref{fig: comparison with tomassi}B).
\end{remark}

\section{Estimation of compositional SDR}\label{sec:compositional-kdr}

This section develops compositional KDR (CKDR), an estimator of compositional SDR that optimizes a kernel-based conditional dependence measure directly over the CDR domain $\cm_{m,d}$ (Sections~\ref{sec:conditional-covariance} and \ref{sec:ckdr-method}). We then show that the CKDR objective is equivalent to a vector-valued kernel ridge regression problem, yielding an intrinsic predictive model (Section~\ref{sec: vv-krr}) and enabling us to propose a successive convex approximation (SCA) algorithm for the nonconvex empirical objective (\cref{sec:optimization-sca}).

\subsection{KDR Criterion via Conditional Covariance Operator}\label{sec:conditional-covariance}

Let $k$ be a positive definite kernel function on a generic domain $\cz$, with the associated reproducing kernel Hilbert space (RKHS) $\ch$, satisfying $k(z, \cdot)\in\ch$ and $\langle f, k(z, \cdot)\rangle_\ch = f(z)$ for all $z\in\cz$ and $f\in\ch$; we use $\langle \cdot, \cdot\rangle_\ch$ and $\|\cdot\|_\ch$ to denote inner product and norm of $\ch$. In this paper, we assume that $k$ is continuous and satisfies $\E_{Z\sim\Q}[k(Z, Z)] < \infty $ for all probability measures $\Q$ on $\cz$. This ensures that $\ch$ is continuously embedded in $L^2(\Q)$, which in turn guarantees the boundedness of the covariance operators introduced below.

A kernel $k$ is said to be \emph{characteristic} if the space $\ch + \R$ is dense in $L^2(\Q)$ for all probability measures $\Q$ on $\cz$. When $\cz$ is compact, we say $k$ is \emph{universal} if $\ch$ is dense in $C(\cz)$, the space of continuous functions on $\cz$ equipped with the uniform norm. It is known that universal kernels are characteristic \citep[Lemma 1]{gretton2012kernel}.

Let $\cx\subseteq\R^d$ denote the domain of predictors and $\cy$ the domain of responses, where $\cy$ is assumed to be a separable metric space. Let $k_\cx: \cx\times\cx\to\R$ and $k_\cy: \cy\times\cy\to\R$ be positive definite kernels with associated RKHSs $\ch_\cx$ and $\ch_\cy$, respectively. Consider a joint random vector $(X, Y)\in\cx\times \cy$, and denote the marginal distributions by $\P_X$ and $\P_Y$. 

The \emph{cross-covariance operator} $\s_{YX}:\ch_\cx\to\ch_\cy$ is defined as a linear operator satisfying
\begin{align}\label{eq: cov-def}
    \begin{split}
        \langle g, \Sigma_{YX}f \rangle_{\ch_\cy} 
    &= \cov \left[f(X), g(Y)\right], \quad \forall f\in\ch_\cx,\ \forall g\in\ch_\cy.
    \end{split}
\end{align}
By definition, its adjoint satisfies $(\s_{YX})^* = \s_{XY}$. When $X = Y$, we write $\s_{XX}$ to denote the covariance operator of $X$. These operators admit a \textit{correlation operator} representation \citep{baker1973joint}: there exists a unique bounded operator $V_{YX}:\ch_\cx\to\ch_\cy$ such that
\begin{equation}\label{eq: correlation}
    \Sigma_{YX} = \Sigma_{YY}^{1/2}V_{YX}\Sigma_{XX}^{1/2}, \quad \|V_{YX}\|\leq 1,\ \text{ and }\ V_{YX} = \Pi_{\clo{\ran}(\Sigma_{YY})}V_{YX}\Pi_{\clo{\ran}(\Sigma_{XX})}, 
\end{equation}
where $\|\cdot\|$ denotes the operator norm, $\clo\ran(\cdot)$ the closure of the range, and $\Pi_W$ the orthogonal projection onto a subspace $W$. 

The \emph{conditional covariance operator} on $\ch_\cy$ is defined as
\begin{equation}\label{eq: condcov-def}
    \s_{YY|X} = \s_{YY} - \s_{YY}^{1/2}V_{YX}V_{XY}\s_{YY}^{1/2}. 
\end{equation}
When $\s_{XX}$ is invertible, this reduces to the familiar expression $\s_{YY} - \s_{YX}\s_{XX}^{-1}\s_{XY}$. For any $g\in\ch_\cy$, the operator $\s_{YY|X}$ characterizes the minimum residual variance: 
\begin{align}
    \langle g, \s_{YY|X} g \rangle_{\ch_\cy} &= \inf_{f\in\ch_\cx} \var(g(Y) - f(X)). \label{eq:residual} 
\end{align}
In particular, if $\{g_i\}_{i=1}^\infty$ is a complete orthonormal system (CONS) of $\ch_\cy$, then the trace $\Tr(\s_{YY|X}) = \sum_{i=1}^\infty\langle g_i, \s_{YY|X} g_i\rangle_{\ch_\cy}$ aggregates such variance over all directions in $\ch_\cy$.

We then consider a target domain $\cz\subseteq\R^m$ with $m\le d$, representing a reduced-dimensional space. Let $k_\cz$ be a kernel defined on $\cz$. For any measurable map $p:\cx\to\cz$, the induced random vector $p(X)$, together with the kernel $k_\cz$, gives rise to operators analogous to those defined for $X$: the cross-covariance operator $\s_{Y p(X)}$, the correlation operator $V_{Y p(X)}$, their adjoints, and the conditional covariance operator $\s_{YY|p(X)}$.

The following proposition establishes the fundamental link between conditional covariance operators and SDR, providing the theoretical basis for the KDR framework.

\begin{proposition}\label{thm:sdr-population-guarantee}
    Suppose that $\ch_\cx$ is dense in $L^2(\P_X)$ (e.g., when $k_\cx$ is universal), and let $p:\cx\to\cz$ be a measurable map. Then,
    $$\s_{YY|p(X)} \succeq \s_{YY|X}, $$
    where $\succeq$ denotes the positive semi-definite order on self-adjoint operators. Moreover, if the kernel $k_\cz$ is characteristic, the following statements hold: 
    \begin{enumerate}
        \item[(i)] If $k_\cy$ is also characteristic, then $\,\s_{YY|p(X)} = \s_{YY|X} \iff Y\ind X\mid p(X).$
        \item[(ii)] If $\cy\subset\ch$ for some separable Hilbert space $\ch$, and $k_\cy$ is the linear kernel $\langle \cdot , \cdot \rangle_\ch$, then
        $\s_{YY|p(X)} = \s_{YY|X} \iff \E[Y|X] = \E[Y|p(X)] \text{ a.s.} $
    \end{enumerate}    
\end{proposition}

Part (i) of Proposition~\ref{thm:sdr-population-guarantee} broadens the scope of the SDR characterization of \citet{fukumizuKernelDimensionReduction2009}, developed for $p$ in the Stiefel manifold $\St(m,d)=\{B\in\R^{m\times d}: BB^\top = I_m\}$, to arbitrary measurable maps $p$. Part (ii) further extends this characterization to SDR for conditional mean; for Euclidean responses, this is practically useful since fixing $k_\cy$ to a linear kernel eliminates kernel selection on the response domain. The operator inequality implies $\Tr(\s_{YY|p(X)}) \geq \Tr(\s_{YY|X})$, with equality if and only if $\s_{YY|p(X)} = \s_{YY|X}$. Hence, the trace $\Tr(\s_{YY|p(X)})$ serves as a natural loss function for SDR.

\begin{remark}\label{rmk:re-embedding}
    We construct the operators differently from earlier KDR work to circumvent orthogonality-specific arguments. Prior approaches use a pullback kernel $k_\cx^p(x, x') = k_\cz(p(x), p(x'))$ on the original domain $\cx$, which induces an RKHS $\ch_\cx^p$ and associated cross-covariance and correlation operators $\s_{YX}^p$ and $V_{YX}^p$, yielding another conditional covariance operator $\s_{YY|X}^p$. With semiorthogonal matrices $B\in\St(m, d)$, theoretical analysis of this pullback formulation proceeds by re-embedding the reduced data $BX$ into $B^\top BX\in\cx$. However, this mechanism fails for general reductions, including our CDR setting: $P^\top PX \notin \Delta^{d-1}$ for general $P\in\cm_{m,d}$ and $X\in\Delta^{d-1}$, hindering a direct extension of KDR beyond the Stiefel manifold. Instead, we build all operators using the RKHS $\ch_\cz$ on the \textit{target domain} $\cz$, yielding the operator $\s_{YY|p(X)}$ that avoids re-embedding. These two constructions differ subtly---$\ch_\cx^p$ is not isomorphic to $\ch_\cz$ when $p:\cx\to\cz$ is not surjective, a situation common for matrices in $\cm_{m,d}$. We show, nonetheless, in Appendix~\ref{sec:supp-kdr-pullback}
    that they carry equivalent distributional information; in particular,  $\s_{YY|X}^p = \s_{YY|p(X)}$ (Lemma~\ref{lemma:condcov-equality}). This target-domain formulation thus preserves equivalent computational properties while facilitating asymptotic analysis on domains with varying matrix rank (Section~\ref{sec:consistency}).
\end{remark}

\subsection{Compositional KDR}\label{sec:ckdr-method}

We now introduce the compositional KDR (CKDR) method for estimating compositional SDR (CSDR) matrices. Let $X\in\cx=\Delta^{d-1}$ and define the target domain as $\cz = \Delta^{m-1}$ with $m \le d$. We assume that the kernel $k_\cz$ on the target simplex is characteristic, which holds, for instance, for the standard Gaussian kernel. By Proposition~\ref{thm:sdr-population-guarantee}, the CKDR population-level criterion is defined as
\begin{equation}\label{eq:ckdr-population}
    \argmin_{P\in\cm_{m,d}} \Tr(\s_{YY|PX}).
\end{equation}
When $k_\cy$ is characteristic and $m \ge \dim(\mathcal{C}_{Y|X})$, where $\mathcal{C}_{Y|X}$ denotes the central compositional subspace defined in \cref{sec:compositional-SDR}, Proposition~\ref{thm:sdr-population-guarantee} ensures that any minimizer $P$ of \eqref{eq:ckdr-population} is a valid CSDR matrix. Moreover, if $k_\cy$ is a linear kernel over a Euclidean or Hilbert space, a similar guarantee holds for recovering the CSDR for conditional mean.

We estimate the minimizer of \eqref{eq:ckdr-population} from an i.i.d. sample $(x_1, y_1), \ldots, (x_n, y_n) \in \Delta^{d-1} \times \cy$ from the joint distribution of $(X, Y)$. Replacing the population covariance in \eqref{eq: cov-def} with its empirical counterpart, the empirical cross-covariance operator $\hatsig_{YX}: \ch_\cx \to \ch_\cy$ is given by
\begin{align*}
    \langle g, \hatsig_{YX} f\rangle_{\ch_\cy} &= \frac{1}{n}\sum_{i=1}^n g(y_i)f(x_i) - \left(\frac{1}{n}\sum_{i=1}^n g(y_i) \right) \left(\frac{1}{n}\sum_{i=1}^n f(x_i)\right)
\end{align*}
for all $g\in\ch_\cy$ and $f\in\ch_\cx$. Let $\hatsig_{XX}$ and $\hatsig_{YY}$ denote the corresponding empirical covariance operators. To ensure operator inversion, we introduce a regularization parameter $\varepsilon_n > 0$. The empirical conditional covariance operator is then defined by
\begin{equation*}%\label{eq: empirical condcov}
    \hatsig_{YY|X}= \hatsig_{YY} - \hatsig_{YX}(\hatsig_{XX} + \varepsilon_n I)^{-1}\hatsig_{XY}. 
\end{equation*}

Given this definition, we estimate a CSDR matrix $P \in \cm_{m,d}$ by replacing $X$ with $PX$ and minimizing the empirical objective $\Tr(\hatsig_{YY|PX})$, computed analogously to \citet{fukumizuKernelDimensionReduction2009} following the discussion in Remark~\ref{rmk:re-embedding}. Specifically, let $K_{PX} = (k_\cz(Px_i, Px_j))_{i,j=1}^n$ be the Gram matrix of the projected data, and define its centered version $G_{PX} = H K_{PX} H$, where $H= I_n - \frac{1}{n} \1_n \1_n^\top$. Similarly, let $K_Y = (k_\cy(y_i, y_j))_{i,j=1}^n$ and $G_Y = H K_Y H$. The empirical conditional trace is then computed as:
\begin{equation}\label{eq:empirical-objective}
    T_n(P) := \Tr(\hatsig_{YY|PX}) = \varepsilon_n \Tr((G_{PX} + n\varepsilon_nI_n)^{-1}G_Y).
\end{equation}
Accordingly, the CSDR estimator is obtained by solving
\begin{equation}\label{eq:empirical-estimate}
    \underset{P\in\cm_{m,d}}{\argmin}\ T_n(P).
\end{equation}
Since the parameter space $\cm_{m,d}$ is compact and the kernel $k_\cz$ is continuous, the empirical objective \eqref{eq:empirical-estimate} admits at least one minimizer, denoted $\wh P_n$. 
In \cref{sec:consistency}, we show that this estimator is consistent, achieving compositional SDR as $n$ tends to infinity.

\begin{remark}\label{rmk:sparsity}
    As illustrated in \cref{fig: comparison with tomassi}, the estimator $\wh P_n$ exhibits emergent sparsity without explicit sparsity-inducing regularization, thereby revealing sufficient amalgamation structure (Remark~\ref{rmk:suff-amalgam}) and enhancing interpretability. This phenomenon can be attributed to the polytopic geometry of the constraint set $\cm_{m,d}$, in analogy to optimization over $\ell^1$-balls that yields sparse solutions \citep{tibshirani1996regression}. \citet{wuSparseTopicModeling2023} also noted a similar property in a different nonlinear optimization over $\cm_{m,d}$, where sparsity was ascribed to the set's inherent nonnegativity constraint. In contrast, classical KDR on the Stiefel manifold does not yield sparse solutions, requiring extra penalties to induce sparsity \citep{liuSparseKernelSufficient2024}.
\end{remark}

\subsection{Intrinsic Predictive Model of CKDR}\label{sec: vv-krr}

Although SDR is a supervised technique, many SDR methods stop at identifying an SDR subspace, without directly addressing the prediction of $Y$. In contrast, CKDR embeds a predictive model within the dimension-reduced domain, where predictions are explicitly computable from samples. This built-in structure enables evaluating reduction quality via predictive performance, supports principled cross-validation for hyperparameters such as the target dimension and the regularization parameter $\varepsilon_n$, and facilitates our optimization algorithm via successive convex approximation (Section~\ref{sec:optimization-sca}).

The intrinsic predictive model of the CKDR framework emerges from a fundamental connection between our trace objective $\Tr(\hatsig_{YY|PX})$ and the objective function of kernel ridge regression (KRR). Specifically, taking a CONS $\{g_i\}_{i=1}^\infty$ of $\ch_\cy$, we can express% the trace as
\begin{align}\label{krr}
    \Tr(\hatsig_{YY|PX}) % &= \sum_{i=1}^\infty \langle g_k, \hatsig_{YY|PX} g_k\rangle_{\ch_\cy} \nonumber \\
    &= \sum_{k=1}^\infty \min_{f \in \ch_\cz,\, c_k\in\R}\bracket*{\frac{1}{n}\sum_{i=1}^n \pr*{ {g_k}(y_i)  - f(Px_i) - c_k }^2 + \varepsilon_n\|f\|_{\ch_\cz}^2 },
\end{align}
where each summand, computed analogously to \eqref{eq:residual}, corresponds to a KRR problem for the scalar response values $(g_k(y_i))_{i=1}^n$ and inputs $(Px_i)_{i=1}^n$. This decomposition, similarly noted by prior studies \citep{fukumizuKernelDimensionReduction2009, liuSparseKernelSufficient2024}, reveals that CKDR implicitly performs an infinite sequence of KRR tasks in the reduced feature space. 
However, for general $k_\cy$, this infinite-sum form leaves it unclear how to obtain explicit predictions.

We show that the expression \eqref{krr} can be \emph{vectorized} by representing the responses $y_i$ via its canonical feature map $k_\cy(y_i, \cdot)\in\ch_\cy$ and invoking the notion of vector-valued RKHS \citep{micchelli2005learning}. Specifically, consider the $\ch_\cy$-valued RKHS $\cg_\cz$ associated with $k_\cz$, defined by a closed linear span of the $\ch_\cy$-valued functions of the form $z \mapsto k_\cz(z, \cdot)\gamma$ for $z\in\cz$ and  $\gamma\in\ch_\cy$. This space satisfies the vector-valued reproducing property 
\[
\langle F, k_\cz(z, \cdot)\gamma\rangle_{\mathcal G_\cz} = \langle F(z), \gamma \rangle_{\ch_\cy} \quad \text{ for all 
 } F\in\cg_\cz, \; z\in\cz, \;\gamma \in \ch_\cy.
 \]
Using this vectorization, we have the following equivalence result:
\begin{theorem}\label{thm:trace=centered-vv-krr}
Define the joint objective function
\begin{equation}\label{eq:krr-objective}
    J_n(P, F, \gamma) = \frac{1}{n}\sum_{i=1}^n \norm*{k_\cy(y_i, \cdot) - F(Px_i) - \gamma}_{\ch_\cy}^2 + \varepsilon_n\norm*{F}_{\mathcal{G}_\cz}^2.
\end{equation}
The equality in \eqref{krr} becomes
\begin{equation}\label{eq:empirical-kdr=vv-krr}
        T_n(P) = \underset{F\in\mathcal{G}_\cz,\,
        \gamma\in\ch_\cy}{\min}\ J_n(P, F, \gamma).
\end{equation}
Therefore, the empirical CKDR estimation \eqref{eq:empirical-estimate} is equivalent to solving
\begin{align}\label{eq:joint-learning}
    \minimize_{P\in\cm_{m,d},\, F\in\mathcal{G}_\cz,\,
        \gamma\in\ch_\cy}\ J_n(P, F, \gamma).
\end{align}
\end{theorem}

That is, $T_n(P)$ corresponds to the vector-valued KRR problem with intercept $\gamma\in\ch_\cy$ on the data $\mathcal{T}_n= \{(Px_i, k_\cy(y_i, \cdot))\}_{i=1}^n \subset\Delta^{m-1}\times \ch_\cy$, and minimizing it amounts to finding $P\in\cm_{m,d}$ that attains the best KRR fit. This illustrates that even when the target dimension is underspecified, $m < \dim \mathcal{C}_{Y|X}$, the CKDR estimator $\wh P_n$ still seeks a reduction that best preserves predictive information. 

As a result of vectorization, we obtain a unique minimizer $(F_P, \gamma_P)$ in an explicit form: writing $\Psi = (k_\cy(y_1, \cdot),\ldots,k_\cy(y_n, \cdot))^\top\in(\ch_\cy)^n$ and $(\alpha_1,\ldots,\alpha_n)^\top = (G_{PX} + n\varepsilon_nI_n)^{-1}H\Psi\in(\ch_\cy)^n$, we have
\begin{equation}\label{eq:vv-krr-closed-form}
    F_P(\cdot) = \sum_{i=1}^n k_\cz(Px_i, \cdot)\alpha_i \quad\text{and}\quad \gamma_P = \frac{1}{n}\sum_{i=1}^n (k_\cy(y_i, \cdot) - F_P(Px_i)).
\end{equation}
Thus, after estimating $\wh P_n$, we can use the pair $(F_{\wh P_n}, \gamma_{\wh P_n})$ as the final predictive model. For any out-of-sample point $(x', y')\in\Delta^{d-1}\times \cy$, the squared prediction error is computed using the reproducing property of $\ch_\cy$ as
\begin{align}\label{eq:out-of-sample-error}
    \begin{split}
        \ce(x', y'\,|\,\mathcal{T}_n) &= \|k_\cy(y', \cdot) - F_{\wh P_n}(\wh P_n x') - \gamma_{\wh P_n}\|_{\ch_\cy}^2, \\
        &= k_\cy(y', y') - 2\,\mbf{k}_{y'}^\top\,\mbf{v}_{x'}\,  + \mbf{v}_{x'} ^\top\, K_Y \mbf{v}_{x'}\,,
    \end{split}
\end{align}
where $\mbf{k}_{y'} =(k_\cy(y_1,y'),\ldots, k_\cy(y_n, y'))^\top \in \R^n$ and $\mbf{v}_{x'}\, = {H}(G_{\wh P_nX}+n\varepsilon_nI_n)^{-1}\wt{\mbf{k}}_{x'} + \frac{1}{n}\mbf 1_n $ with $\wt{\mbf{k}}_{x'}^\top = \left(k_\cz(\wh P_nx_1, \wh P_nx'), \ldots, k_\cz(\wh P_nx_n, \wh P_nx')\right) - \frac{1}{n}\mbf 1^\top K_{\wh P_nX}$. The sum of such errors over a test dataset provides a natural measure of CKDR's generalization performance, which we use for hyperparameter selection via cross-validation.
Moreover, when responses are real-valued and $k_\cy$ is linear, the $\ch_\cy$-valued predictions $F_{\wh P_n}(\wh P_n x') + \gamma_{\wh P_n}$ translate to $\cy$-valued predictions since $\ch_\cy = \R$; these downstream predictions demonstrate competitive and often best performance relative to competitors in our experiments (\cref{sec: experiments}).

\subsection{Optimization of CKDR via Successive Convex Approximation}\label{sec:optimization-sca}

The empirical CKDR objective $T_n(P)$ is nonconvex in $P\in\cm_{m,d}$, due to the nonlinear dependence of the Gram matrix $G_{PX}$ on~$P$ and the invariance of the objective under row permutations of~$P$. To address this, we develop an algorithm based on successive convex approximation over the convex domain $\cm_{m,d}$, using the Gaussian kernel $k_\cz(z,z') =\exp (-\|z - z'\|^2/2\sigma^2)$ on the target simplex $\Delta^{m-1}$.

Given the current CDR matrix $P_t\in\cm_{m,d}$ at iteration $t$, a natural first attempt at building a convex surrogate for $T_n$ is to linearize the Gaussian kernel at $P_t$: writing $z_i^t = P_tx_i$,
$$ k_\cz(Px_i, Px_j) \approx k_\cz(z_i^t, z_j^t)\Bigl(1 - \frac{1}{\sigma^2}(z_i^t - z_j^t)^\top (P-P_t)(x_i - x_j)\Bigr). $$
Substituting this linearization into the Gram matrix $G_{PX}$ yields a first-order approximation $\wt G_{PX}$, which has often been considered in RKHS-based optimization problems \citep{allen2013automatic}. However, in KDR-based objectives like \eqref{eq:empirical-objective}, the matrix inverse $(\wt G_{PX} + n\varepsilon_n I_n)^{-1}$ may become ill-defined because the linearized matrix $\wt G_{PX}$ need not be positive semidefinite, rendering direct kernel linearization unsuitable in our problem.

We circumvent this difficulty by leveraging the joint learning formulation \eqref{eq:joint-learning}.
Let $(F_t, \gamma_t) := (F_{P_t}, \gamma_{P_t})$ denote the unique vector-valued KRR solution at $P_t$, with $F_t(\cdot) = \sum_{j=1}^n k_\cz(z_j^t, \cdot)\alpha_j^t$ as defined in \eqref{eq:vv-krr-closed-form}. Since $(F_t, \gamma_t)$ minimizes $J_n(P_t, \cdot\,,\cdot)$, the first-order optimality conditions for the partial Fr\'echet derivatives $D_FJ_n(P_t, F_t, \gamma_t) = D_\gamma J_n(P_t, F_t, \gamma_t) = 0$ are satisfied. Then, by the 
envelope theorem \citep[Remark~4.14]{bonnansPerturbationAnalysisOptimization2000},
\begin{equation*}
    \nabla T_n(P_t) = D_P J_n(P_t, F_t, \gamma_t).
\end{equation*}
This enables first-order approximating the objective $T_n$ at $P_t$ using the partial gradient $D_P J_n(P_t, F_t, \gamma_t)$, which regards the function $F_t= F_{P_t}$ as fixed. 

Given $J_n(P, F_t, \gamma_t) = \frac{1}{n}\sum_i\|k_\cy(y_i, \cdot) - F_t(Px_i) - \gamma_t\|_{\ch_\cy}^2 + \varepsilon_n\|F_t\|_{\cg_\cz}^2$, we consider an affine approximation of $P\mapsto F_t(Px_i)$ using the linearization of Gaussian kernel $k_\cz(z_j^t, \cdot)$:
\begin{equation*}
    \wt{F}_{t,i}(P) = F_t(z_i^t) - \frac{1}{\sigma^2}\sum_{j=1}^n k_\cz(z_i^t, z_j^t)(z_i^t - z_j^t)^\top (P - P_t) x_i\, \alpha_j^t.
\end{equation*}
Let $R_t = (G_{P_tX} + n\varepsilon_n I_n)^{-1}$ and $A_\alpha^t = R_t G_Y R_t^\top$. For each sample $i$, define $V_i^t = (z_i^t - z_1^t, \ldots, z_i^t - z_n^t)\in\R^{m\times n}$, $C_i^t = \sigma^{-2}\operatorname{diag}(k_\cz(z_i^t, z_1^t), \ldots, k_\cz(z_i^t, z_n^t))\in\R^{n\times n}$, and 
$$
S_i^t = V_i^t C_i^t A_\alpha^t C_i^t (V_i^t)^\top\in\R^{m\times m}, \qquad a_i^t = n\varepsilon_n\, V_i^t C_i^t A_\alpha^t e_i\in\R^m,
$$
where each $S_i^t$ is positive semidefinite. Then, by substituting the affine approximation $\wt{F}_{t, i}(P)$ into $J_n$, we define the convex quadratic surrogate $Q_t(P)$ as
\begin{align}\label{eq:surrogate-quadratic}
    \begin{split}
        Q_t(P) &:= \frac{1}{n}\sum_{i=1}^n \|k_\cy(y_i,\cdot) - \wt{F}_{t,i}(P) - \gamma_t\|_{\ch_\cy}^2 + \varepsilon_n \|F_t\|_{\cg_\cz}^2 \\
        &= \mathrm{const} + \frac{2}{n}\sum_{i=1}^n (a_i^t)^\top (P - P_t) x_i + \frac{1}{n}\sum_{i=1}^n x_i^\top (P - P_t)^\top S_i^t (P - P_t) x_i,
    \end{split}
\end{align}
where the const term is independent of $P$.
By the envelope argument above, $\nabla Q_t(P_t) = \nabla T_n(P_t) = \frac{2}{n}\sum_{i=1}^n a_i^t x_i^\top$, ensuring $Q_t$ as a valid first-order convex surrogate for $T_n$ at $P_t$.

Using this surrogate construction, we propose to solve the CKDR optimization \eqref{eq:empirical-estimate} via a successive convex approximation (SCA) algorithm, outlined in Algorithm~\ref{alg:sca}. At each iteration $t$, we minimize the convex surrogate $Q_t$ over $\cm_{m,d}$ using Nesterov's accelerated projected gradient descent with the column-wise simplex projection \citep{duchi2008efficient}, and then apply Armijo backtracking to the original objective $T_n$.

\begin{algorithm}[t]
\spacingset{1}
\caption{Successive Convex Approximation for CKDR}\label{alg:sca}
\begin{algorithmic}[1]
\REQUIRE Data $(x_i, y_i)_{i=1}^n\subset\Delta^{d-1}\times\cy$; target dimension $m$; regularization $\varepsilon_n$; Gaussian kernel bandwidth $\sigma$; Armijo parameters $c = 10^{-4}$, $\beta = 0.5$; tolerance $\epsilon_{\mathrm{tol}}$.
\STATE Initialize $P_0\in\cm_{m,d}$ (we draw each column from $\operatorname{Dir}(m^2,\ldots,m^2)$).
\FOR{$t = 0, 1, 2, \ldots$}
\STATE Build the surrogate $Q_t(\cdot)$ using $P_t$ via \eqref{eq:surrogate-quadratic}: compute $V_i^t$, $C_i^t$, $S_i^t$, $a_i^t$ for each $i$.
\STATE Compute $\nabla T_n(P_t) = \frac{2}{n}\sum_{i=1}^n a_i^t x_i^\top$.
\STATE Minimize $Q_t(\cdot)$ over $\cm_{m,d}$ by accelerated projected gradient descent; let $\wt{P}_t$ be the result and set $D_t = \wt{P}_t - P_t$.
\IF{$\min\big\{\|D_t\|_F, Q_t(P_t) - Q_t(\widetilde{P}_t)\big\} \le \epsilon_{\mathrm{tol}}$}
\STATE \textbf{stop}
\ENDIF
\STATE Choose $\lambda_t\in\{1, \beta, \beta^2, \ldots\}$ as the first value satisfying
$$T_n(P_t + \lambda_t D_t) \le T_n(P_t) + c\,\lambda_t\langle \nabla T_n(P_t), D_t\rangle_F.$$

\STATE Set $P_{t+1} = P_t + \lambda_t D_t$.
\ENDFOR
\end{algorithmic}
\end{algorithm}

We establish the stationary convergence result below using similar proof structures to standard SCA analyses (e.g.,  \citealt{scutariParallelDistributedMethods2017}). While such analyses often require strongly convex surrogates, our compact convex CDR domain $\cm_{m,d}$ and the Armijo line search facilitate a convergence result that does not require $Q_t$ to be strongly convex.

\begin{proposition}\label{prop:sca-convergence}
    Suppose that each inner subproblem $\min_{P\in\cm_{m,d}} Q_t(P)$ is solved exactly. Then every accumulation point of $\{P_t\}_{t\ge 0}$ is first-order stationary for $T_n$ over $\cm_{m,d}$.
\end{proposition}

\begin{remark}\label{rmk:pgd}
    SCA approaches for optimizing nonconvex functions over convex domains have demonstrated numerous advantages, including convergence speed \citep{scutariDecompositionPartialLinearization2014, scutariParallelDistributedMethods2017}, compared to simpler alternatives such as projected gradient descent (PGD). Although PGD has demonstrated good empirical performance in related KDR contexts \citep{chen2017kernel, chenTheoryFeatureLearning2026}, we empirically found that, in most cases, the proposed SCA algorithm for CKDR converges to lower objective values than PGD.
\end{remark}

\section{Consistency of CKDR Estimator}\label{sec:consistency}

This section establishes the consistency of our CKDR estimator $\wh P_n$ for the central compositional subspace $\mathcal{C}_{Y|X}$ when $k_\cy$ is characteristic, by developing an analysis that accommodates the varying-rank structure of the matrix domain $\cm_{m,d}$.
This rank variability invalidates earlier uniform convergence arguments \citep{fukumizuKernelDimensionReduction2009}, as the population objective $T(P)=\Tr(\s_{YY|PX})$ is discontinuous when a sequence of matrices converges to a lower-rank limit (see Appendix~\ref{sec:supp-counterexample}),
and therefore cannot be uniformly approximated by the globally continuous, $\varepsilon_n$-regularized empirical objective $T_n(P)$. Consequently, we develop a new asymptotic analysis that simultaneously (i) rules out rank collapse of $\wh P_n$ below $\dim\mathcal{C}_{Y|X}$, and (ii) quantifies convergence between subspaces of potentially different dimensions, since $\rk(\wh P_n)$ can vary with $n$. An analogous conclusion holds for the mean subspace $\mathcal{C}_{\E[Y|X]}$ when $k_\cy$ is linear; for brevity, we focus here on the characteristic case.

Let $\Pi_V$ denote the orthogonal projection matrix onto a subspace $V\subseteq\R^d$. Let $\Gr(k, d)$ denote the Grassmann manifold of $k$-dimensional subspaces of $\R^d$, and let $\Gr^\1(k, d)$ denote the subset of subspaces that contain $\1_d$.

We list the following assumptions for our asymptotic analysis, which parallel common assumptions in KDR but avoid re-embedding to $\cx=\Delta^{d-1}$ as noted in Remark~\ref{rmk:re-embedding}:
\begin{assumption}\label{assumption:kernels-and-dimension}
    (a) The kernels $k_\cy$ on $\cy$ and $k_\cz$ on $\cz=\Delta^{m-1}$ are characteristic, and (b) $m\ge\dim\mathcal{C}_{Y|X}$.
\end{assumption}
Although $\dim\mathcal{C}_{Y|X}$ in Assumption~\ref{assumption:kernels-and-dimension}(b) is unknown in practice, it relaxes the exact structural dimension specification $m=\dim\mathcal{C}_{Y|X}$ required by most consistency theory of Euclidean SDR methods \citep{li2018sufficient}. Under Assumption~\ref{assumption:kernels-and-dimension}, Proposition~\ref{thm:sdr-population-guarantee} ensures the population objective $T$ attains its global minimum at some $P^\star\in\cm_{m,d}$ with $\row(P^\star)\supseteq\mathcal{C}_{Y|X}$.

\begin{assumption}\label{assumption:modified-continuity}
    For any bounded continuous function $g$ on $\cy$, the mapping
    $$V\mapsto \E[\E[g(Y)|\Pi_VX]^2] $$
    is continuous on $\Gr^\1(k,d)$ for every $k\le m$.
\end{assumption}

\begin{assumption}\label{assumption: Lipschitz}
    There exists a measurable function $\varphi:\Delta^{d-1}\to\R$ with $\E[\varphi(X)^2] < \infty$ such that the Lipschitz condition
    $$\norm*{k_\cz(P_1x, \cdot) - k_\cz(P_2x, \cdot)}_{\ch_\cz} \leq \varphi(x)\, \|P_1 - P_2\| $$
    holds for all $x\in\Delta^{d-1}$ and $P_1, P_2\in\cm_{m,d}$, where $\|\cdot\|$ is the operator norm.
\end{assumption}

Assumption~\ref{assumption:modified-continuity} modifies the continuity condition (A-1) of \citet{fukumizuKernelDimensionReduction2009} to derive the rank-wise continuity of $T$ (see Remark~\ref{rmk:consistency-proof}), addressing its discontinuity at rank changes. As shown therein, a mild sufficient condition for Assumption~\ref{assumption:modified-continuity} is when $X$ has a density on $\Delta^{d-1}$ and the conditional distribution $F_{Y|X}(y|x)$ is continuous in $x$.
Assumption~\ref{assumption: Lipschitz} is a Lipschitz condition satisfied by common kernels $k_\cz$, including the Gaussian and the rational quadratic kernel, and ensures uniform control of empirical cross-covariance operators.

To compare subspaces with possibly different dimensions, we employ the \textit{chordal distance} introduced in \citet{yeSchubertVarietiesDistances2016}. For subspaces $V$ and $W$ of $\R^d$ with dimensions $k$ and $l$, respectively, the squared distance is defined as:
\begin{equation}\label{eq:subsp-distance}
    \rho^2(V, W) = (\|\Pi_V - \Pi_W\|_F^2 - |k-l|) / 2\min(k, l),
\end{equation}
which ranges from 0 to 1. This distance vanishes when one subspace is contained in the other; that is, $\rho(V, W) = 0$ if and only if either $V\subseteq W$ or $W\subseteq V$. When $k=l$, $\rho$ reduces to a standard subspace metric on $\Gr(k, d)$.

We now state our main result below. Our asymptotic analysis involves two main steps: (i) ruling out rank deficiency $\rk(\wh P_n) < \dim\mathcal{C}_{Y|X}$, which prevents proper inclusion $\row(\wh P_n)\subsetneq \mathcal{C}_{Y|X}$; and (ii) establishing the convergence $\rho(\row(\wh P_n), \mathcal{C}_{Y|X}) \to 0$. These two steps imply that $\row(\wh P_n)$ asymptotically \textit{contains} the central compositional subspace $\mathcal{C}_{Y|X}$, thereby guaranteeing compositional SDR.

\begin{theorem}\label{thm: consistency}
    Suppose that the regularization parameter $\varepsilon_n$ in \eqref{eq:empirical-estimate} satisfies 
    \begin{equation}\label{eq: epsilon condition}
        \varepsilon_n \to 0 \quad \text{and}\quad n^{1/2}\varepsilon_n\to\infty \quad \text{as}\quad n\to\infty.
    \end{equation} 
    Under Assumptions \ref{assumption:kernels-and-dimension}, \ref{assumption:modified-continuity}, and \ref{assumption: Lipschitz}, for every positive number $\delta > 0$, we have
    $$\lim_{n\to\infty}\P\left(
    \rk(\wh P_n)\ge \dim \mathcal{C}_{Y|X}\ \wedge\ \rho(\row(\wh P_n), \mathcal{C}_{Y|X}) < \delta\right) = 1.$$
\end{theorem}

As a corollary, we guarantee the exact recovery of $\mathcal{C}_{Y|X}$ when $m=\dim\mathcal{C}_{Y|X}$ is specified.

\begin{remark}\label{rmk:consistency-proof}
    We prove Theorem~\ref{thm: consistency} in two steps. First, extending prior KDR results, we show that $T(\wh P_n) \to T(P^\star)$ in probability. Second, we show that $T(P^\star)$ is separated from the infimum of $T$ over each of two excluded subsets of $\cm_{m,d}$: matrices with $\rk(P) < \dim\mathcal{C}_{Y|X}$, and matrices satisfying $\rho(\row(P), \mathcal{C}_{Y|X}) \ge \delta$ for a fixed $\delta>0$. These positive gaps imply that $\wh P_n$ avoids both subsets with probability tending to one, thereby proving the theorem.    
    While a related lower-rank separation result is suggested by \citet[Lemma~12]{chenTheoryFeatureLearning2026}, their $L^2$ weak-limit argument has a counterexample; see Remark~\ref{rmk:counterexample-previous} in the Appendix. Instead, we control the infima over both excluded subsets by combining rank-wise continuity of $T$ on $\cm_{m,d}^{(k)}$ with the surjective row-space map $\cm_{m,d}^{(k)} \to \Gr^\1(k,d)$ from Lemma~\ref{lemma:1subspace}, which addresses the non-compactness of $\cm_{m,d}^{(k)}$. 
    See Appendix~\ref{sec:supp-consistency-proof} for full details.
\end{remark}

\section{Simulations and Real Data Analysis}\label{sec: experiments}

In this section, we assess the utility and performance of our method via simulations and real-world microbiome datasets. We consider both binary and univariate continuous responses.

\subsection{Experimental Settings}\label{sec:experimental-settings}

CKDR uses the Gaussian kernel $k_\cz(z,z') =\exp (-\|z - z'\|^2/2\sigma^2)$ on the target simplex and is implemented by the SCA algorithm (Section~\ref{sec:optimization-sca}). For real-valued responses, we use the linear kernel $k_\cy(y, y') = yy'$. For binary responses, we encode $\cy=\{-1,1\}$, making the linear kernel $k_\cy$ characteristic on $\cy$; in this case, CKDR estimates the central compositional subspace $\mathcal{C}_{Y|X}$, whereas it targets the mean subspace $\mathcal{C}_{\E[Y|X]}$ for continuous responses. Hyperparameters are selected via 5-fold cross-validation using the test error in \cref{eq:out-of-sample-error}. We set the bandwidth $\sigma$ to the median pairwise distance among $\{\|x_i-x_j\|\}_{i<j}$. The regularization parameter $\varepsilon_n$ is chosen from 10 values equally spaced on the log scale in $[10^{-4}, 10^{-2}]$. To assess the effect of the target dimension $m$, we consider two scenarios: (a) CKDR$^*$, where $m\in\{ 3, 4, 5, 6, 7\}$ is tuned jointly with $\varepsilon_n$, and (b) CKDR-$m$, where $m$ is fixed a priori.

We also compare the performance of the intrinsic predictive model of CKDR against existing competitors. Under the linear kernel $k_\cy$, the model yields real-valued predictions $\hat y \in \ch_\cy=\R$, which are used for evaluation. For binary responses, we apply $\sign(\hat y)\in\{-1, 1\}$. Competitors include the log-contrast model with $\ell_1$-penalty (LC-Lasso) \citep{linVariableSelectionRegression2014, luGeneralizedLinearModels2019}, KRR or support vector machine with a Gaussian kernel after centered log-ratio (clr) transformation (clr-Kernel), random forest after clr transformation (clr-RF), and relative-shift regression with equi-sparsity penalty (RS-ES) \citep{liItsAllRelative2023}, which is included for continuous responses. For the three log-ratio-based methods, zeros in $x$ are replaced by $.5 x_{\min}$, where $x_{\min}$ denotes the smallest positive entry in $x$. We implement LC-Lasso using the Python \texttt{c-lasso} library \citep{simpsonClassoPythonPackage2021} for continuous responses and R code from \citet{susin2020variable} for binary responses, tuning the regularization parameter over 30 values equally spaced on the log scale in $[0.001, 1]$. For clr-Kernel, we use the Gaussian kernel with bandwidth $\sigma$ set to the median pairwise distance of clr-transformed data; the ridge parameter and SVM cost parameter are tuned over ten values equally spaced on the log scale in $[0.01, 1]$ and $[1, 100]$, respectively. The clr-RF uses 100 decision trees via \texttt{scikit-learn} \citep{pedregosa2011scikit}. Finally, RS-ES uses the provided MATLAB code.

\subsection{Simulations}\label{sec: simulations}

In simulations, we assess the performance of CKDR in terms of both compositional SDR estimation and prediction. For sample sizes $n\in\{200, 500, 1000\}$, we generate $d=100$ compositional covariates by drawing $n$ vectors from a logistic Gaussian distribution with mean zero and covariance $\Sigma = (0.2^{|i-j|})_{i, j =1}^d$, truncating the lower 50\% of the entries to zero, and subsequently renormalizing to obtain compositions with structural zeros. 

The true underlying structure consists of three amalgamated variables: 
$Z_1 = \sum_{j=1}^{20} X_j,$ $Z_2 = \sum_{j=21}^{50} X_j, \; \text{and }  Z_3 = \sum_{j=51}^{100} X_j$. 
Responses are then generated from two regression and two binary classification models:
\begin{itemize}
    \item[i.] $\ Y = -5Z_1 + 4Z_3 + 0.1\epsilon$
    \item[ii.] $\ Y = 3\cos(Z_1) + {Z_3^2} /{(Z_2 + 0.01)} + 0.1\epsilon$
    \item[iii.] $\ Y = \sign(5Z_2 - 3Z_3 + 0.1\epsilon) $
    \item[iv.] $\ Y = \sign(3Z_1^2 + 4Z_2^2 - 2Z_3^2 + 0.1\epsilon) $
\end{itemize}
where $\epsilon \sim N(0,1)$. In all cases, the central compositional subspace $\mathcal{C}_{Y|X}$ coincides with the mean subspace $\mathcal{C}_{\E[Y|X]}$. The subspace dimension is $m^\star=2$ for (i) and (iii), and $m^\star=3$ for (ii) and (iv). For each setting, the averaged performance over 100 repetitions is recorded; hyperparameters are tuned in the first run and fixed for all subsequent repetitions.

Since no Euclidean SDR methods directly estimate the central compositional subspace $\mathcal{C}_{Y|X}$,
we compare the estimation performance of CKDR against methods for amalgamation: RS-ES by \citet{liItsAllRelative2023} in the regression settings (i) and (ii), and ``Amalgam'' by \cite{quinnAmalgamsDatadrivenAmalgamation2020} in the classification settings (iii) and (iv). RS-ES fits a linear model $Y = \sum_{j=1}^d \beta_j X_j$, from which the fitted coefficients $\hat\beta_j$ are used to construct a rank-2 CDR matrix $\wh P_n$ as in \eqref{eq:rs-es}. Amalgam searches a $K$-part amalgamation via a genetic algorithm with a log-ratio-based criterion after zero replacement. The resulting amalgamation yields a rank-$K$ binary CDR matrix $\wh P_n$; we set $K = 3$ (the true value) to favor this method. For CKDR, we consider the oracle-dimension setting by fixing $m=m^\star$ (CKDR-$m^\star$), as well as CKDR$^*$ in which $m$ is also cross-validated.

For the evaluation metric, we use the distance $\rho(\row(\wh P_n), \mathcal{C}_{Y|X})$ (see \cref{sec:consistency}) to assess the SDR convergence in the sense of inclusion $\row(\wh P_n)\supseteq \mathcal{C}_{Y|X}$. We further examine whether the true amalgamation structure is recovered by clustering the columns of $\wh P_n$ in the simplex using the $k$-quantiles clustering \citep{wei2017k-quantiles} with $k=3$, and then computing the adjusted Rand index (ARI) relative to the true variable amalgamation.

\begin{table}[t]
    \centering
    \spacingset{1}
    \caption{Simulation results on estimation accuracy for SDR and true amalgamation, with standard errors in parentheses. Bold-faced numbers indicate the best result for each setting.}
    \label{tab:simulation-subspace}
    \begin{scriptsize}
    \begin{tabular*}{\columnwidth}{@{\extracolsep{\fill}}llcccccc}
    \toprule
     &  & \multicolumn{3}{c}{$\rho(\row(\wh P_n), \mathcal{C}_{Y|X}) \times 100$} & \multicolumn{3}{c}{ARI $\times 100$} \\
    \cmidrule(lr){3-5} \cmidrule(lr){6-8}
    Setting& Method & $n=200$ & $n=500$ & $n=1000$ & $n=200$ & $n=500$ & $n=1000$ \\
    \midrule
     \multirow{3}{*}{(i)} & CKDR-$m^\star$ & 13.1 (0.4) & 7.5 (0.6) & \textbf{4.2 (0.0)} & 98.0 (0.7) & 96.3 (1.4) & 98.8 (0.8) \\
      & CKDR$^*$ & \textbf{11.9 (0.4)} & \textbf{6.4 (0.1)} & 4.4 (0.0) & 53.6 (2.3) & 37.7 (1.5) & 49.9 (2.5) \\
      & RS-ES & 12.8 (0.1) & 6.5 (0.1) & 4.3 (0.0) & \textbf{99.5 (0.1)} & \textbf{97.6 (1.2)} & \textbf{98.8 (0.8)} \\
    \addlinespace
     \multirow{3}{*}{(ii)} & CKDR-$m^\star$ & 56.3 (0.4) & \textbf{44.1 (0.7)} & \textbf{31.4 (0.8)} & 62.0 (1.7) & 82.0 (2.4) & 88.2 (2.2) \\
      & CKDR$^*$ & \textbf{55.9 (0.4)} & 45.0 (0.7) & 32.7 (0.7) & \textbf{62.7 (1.5)} & \textbf{84.1 (2.2)} & \textbf{91.7 (1.8)} \\
      & RS-ES & -- & -- & -- & 55.9 (1.4) & 68.3 (2.1) & 74.6 (2.3) \\
    \addlinespace
     \multirow{3}{*}{(iii)} & CKDR-$m^\star$ & \textbf{35.7 (0.3)} & \textbf{18.7 (0.2)} & \textbf{12.2 (0.1)} & \textbf{45.5 (0.7)} & \textbf{73.3 (0.9)} & \textbf{93.6 (0.5)} \\
      & CKDR$^*$ & 56.4 (1.0) & 29.6 (1.7) & 12.3 (0.7) & 41.1 (1.0) & 69.6 (1.3) & 79.1 (1.4) \\
      & Amalgam & 56.2 (0.4) & 42.5 (0.5) & 35.8 (0.4) & 20.2 (0.5) & 34.2 (0.6) & 40.6 (0.4) \\
    \addlinespace
     \multirow{3}{*}{(iv)} & CKDR-$m^\star$ & \textbf{64.4 (0.2)} & \textbf{56.8 (0.3)} & \textbf{49.9 (0.6)} & \textbf{42.6 (0.8)} & \textbf{66.9 (1.1)} & \textbf{76.9 (1.4)} \\
      & CKDR$^*$ & 74.8 (0.3) & 68.0 (0.5) & 53.1 (0.5) & 42.2 (0.9) & 63.1 (1.2) & 76.3 (1.4) \\
      & Amalgam & 75.1 (0.2) & 70.4 (0.2) & 66.6 (0.2) & 17.4 (0.5) & 25.3 (0.6) & 34.5 (0.9) \\
    \bottomrule
    \end{tabular*}
    \end{scriptsize}
\end{table}

The results shown in \cref{tab:simulation-subspace} indicate that the oracle CKDR-$m^\star$ gives the most consistently accurate recovery of both the central compositional subspace and the true amalgamation structure. 
While CKDR$^*$ remains overall highly competitive in subspace recovery, its ARI deteriorates in settings (i) and (iii), which may reflect cross-validation’s selection of a dimension $m$ larger than the true $m^\star=2$. 
RS-ES achieves the highest ARI in the correctly specified linear setting (i), but its performance degrades in the nonlinear setting (ii). Amalgam fails to recover the true sufficient amalgamation in both classification settings. Consequently, although CKDR does not explicitly target discrete amalgamation---unlike RS-ES and Amalgam---it recovers the underlying sufficient amalgamation structure with good overall accuracy. % in the nonlinear and classification settings.

\begin{table}[t]
\spacingset{1}
\centering
\caption{Simulation results on prediction performance, measured by MSE for settings (i) and (ii), and MCR for (iii) and (iv). Standard errors are given in parentheses.}
\label{tab:simulation_predictions}
\scriptsize
\begin{tabular*}{\columnwidth}{@{\extracolsep{\fill}}lllcccccc}
\toprule
Metric & Setting & $n$ & CKDR-$m^\star$ & CKDR$^*$ & LC-Lasso & clr-Kernel & clr-RF & RS-ES \\
\midrule
 \multirow{7}{*}{MSE} & \multirow{3}{*}{(i)} & 200 & .026 (.001) & .030 (.001) & .032 (.000) & .125 (.002) & .316 (.004) & \textbf{.020 (.000)} \\
  &  & 500 & .014 (.001) & .020 (.000) & .019 (.000) & .090 (.001) & .281 (.002) & \textbf{.012 (.000)} \\
  &  & 1000 & .012 (.000) & .014 (.000) & .017 (.000) & .079 (.000) & .262 (.001) & \textbf{.011 (.000)} \\
 \addlinespace
  & \multirow{3}{*}{(ii)} & 200 & \textbf{.077 (.004)} & .081 (.006) & .164 (.005) & .185 (.005) & .345 (.008) & .130 (.004) \\
  &  & 500 & \textbf{.031 (.001)} & .032 (.002) & .107 (.002) & .145 (.002) & .315 (.004) & .096 (.002) \\
  &  & 1000 & \textbf{.022 (.001)} & .023 (.001) & .099 (.002) & .136 (.003) & .306 (.004) & .091 (.002) \\
\midrule
 \multirow{7}{*}{MCR} & \multirow{3}{*}{(iii)} & 200 & \textbf{.156 (.003)} & .161 (.003) & .229 (.004) & .224 (.004) & .338 (.003) & -- \\
  &  & 500 & \textbf{.088 (.001)} & .090 (.002) & .191 (.002) & .180 (.002) & .290 (.002) & -- \\
  &  & 1000 & \textbf{.067 (.001)} & .069 (.001) & .154 (.001) & .155 (.001) & .258 (.002) & -- \\
 \addlinespace
  & \multirow{3}{*}{(iv)} & 200 & .180 (.003) & \textbf{.179 (.003)} & .286 (.004) & .256 (.004) & .361 (.004) & -- \\
  &  & 500 & \textbf{.114 (.002)} & .116 (.002) & .212 (.002) & .201 (.002) & .315 (.003) & -- \\
  &  & 1000 & .100 (.001) & \textbf{.098 (.001)} & .183 (.001) & .178 (.001) & .284 (.001) & -- \\
\bottomrule
\end{tabular*}
\end{table}

Next, we assess the predictive performance using the intrinsic predictive model. We compute the mean squared error (MSE) for settings (i) and (ii) and the misclassification rate (MCR) for settings (iii) and (iv), based on independent test data of size $n$. \cref{tab:simulation_predictions} reports the results averaged over 100 repetitions. In the linear setting (i), CKDR-$m^\star$, CKDR$^*$, and RS-ES perform comparably, with RS-ES achieving the lowest MSE owing to its correct model specification. Across all settings, however, CKDR-$m^\star$ and CKDR$^*$ consistently deliver strong performance, outperforming most of the competing methods. 

\subsection{Analysis of Real Microbiome Data}\label{sec:real-data-experiments}

We apply our method to two real-world microbiome datasets to demonstrate its utility across different tasks: the Crohn's disease (CD) study \citep{gevers2014treatment} for binary classification, and the vaginal microbiome study \citep{ravel2011vaginal} for continuous response regression.

We first apply our method to the CD study \citep{gevers2014treatment} to understand the association between CD status and ileum microbiome compositions. The data, available at ML Repo \citep{vangayMicrobiomeLearningRepo2019}, comprises treatment-naive pediatric patients with newly diagnosed CD. After removing taxa observed in fewer than five samples, we obtain $d=194$ microbial taxa counts at the highest available taxonomic resolution across $n=140$ subjects, with 82\% of counts equal to zero. These counts are normalized to compositions. The dataset includes 78 CD patients and 62 healthy controls, forming the binary response variable.

\begin{figure}%[t]
    \centering
   \includegraphics[width=.8\textwidth]{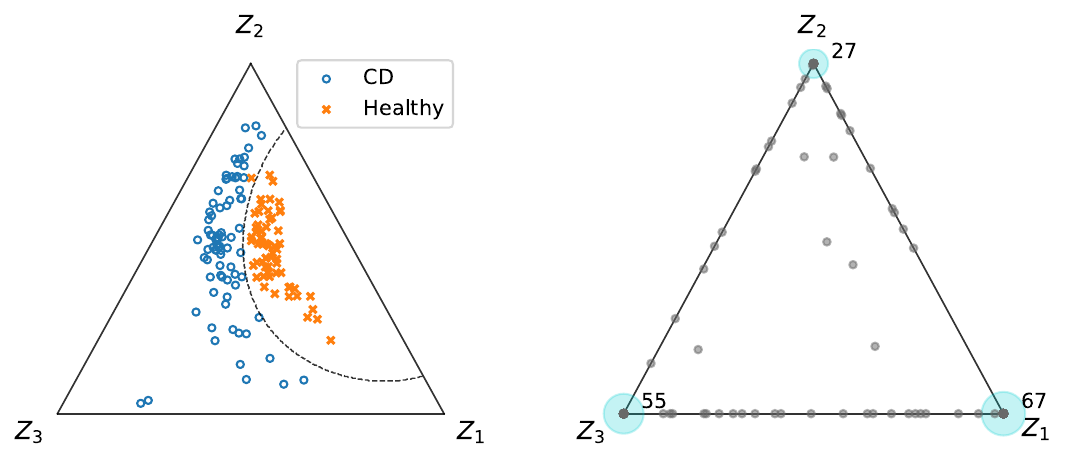}
   \spacingset{1}
    \caption{
    Dual visualization of the ileum microbiome dataset from CKDR-3. \textit{Left}: data projected onto $\Delta^2$, with the dashed curve showing the decision boundary estimated from the intrinsic predictive model. 
    \textit{Right}: variable allocation plot illustrating the contributions of the original variables to the dimension-reduced predictors. Cyan bubbles mark clusters of variables near vertices, with their sizes and labels indicating the cluster counts. 
    }
    \label{fig: ileum}
\end{figure}

\cref{fig: ileum} presents the dual visualization of the ileum microbiome data using CKDR. The left panel shows the projection onto $\Delta^2$, where CD and healthy subjects are clearly separated by a nonlinear decision boundary (dashed curve) derived from the intrinsic predictive model. The discrimination is primarily driven by the subcomposition $(z_1, z_3)$: higher relative abundance of $z_3$ over $z_1$ corresponds to CD, whereas the reverse indicates healthy status, with $z_2$ contributing little. The right panel displays the variable allocation plot, which reveals pronounced emergent sparsity: most columns $P_j = (p_{1j}, p_{2j}, p_{3j})^\top$ of $\wh P_n$ lie on the simplex boundary, with 77\% clustering near vertices ($\max_k p_{kj} > 0.9$). Within the subcomposition $(z_1, z_3)$, columns near the left and right edges correspond to higher abundance of $z_3$ and $z_1$, respectively. We interpret the left-edge cluster as CD-associated and the right-edge cluster as health-associated, with representative genera listed in \cref{tab:ileum_taxa}. 
These data-driven findings align with existing literature. Genera such as \textit{Haemophilus} and \textit{Lachnoclostridium}, which have been reported as enriched in CD patients \citep{metwalyIntegratedMicrobiotaMetabolite2020, alsulaimanGutMicrobiotaAnalyses2023}, cluster near the CD-associated edge. In contrast, short-chain fatty-acid (SCFA)--producing bacteria, including \textit{Blautia}, \textit{Eubacterium}, and \textit{Roseburia}, cluster near the healthy-associated edge, consistent with reports of their depletion in CD and their protective role in delaying disease progression \citep{zhangShortchainFattyAcids2023, maGutMicrobiotaEarly2022}.

\begin{table}%[t]
\centering
\caption{
Top 5 frequent genera (with the species counts in parentheses) near the left edge (CD) and the right edge (Healthy) of the variable allocation plot in \cref{fig: ileum}. The proximity to each edge is defined as $\{j\in[d]: p_{3j} > 10\cdot p_{1j}\}$ for CD and $\{j\in[d]: p_{1j} > 10\cdot p_{3j}\}$ for healthy, where $P_j = (p_{1j}, p_{2j}, p_{3j})^\top$.
}
\label{tab:ileum_taxa}
\footnotesize
\begin{tabularx}{\columnwidth}{l X} % X automatically adjusts the width
\toprule
CD & \textit{Bacteroides} (7); \textit{Dialister} (3); \textit{Haemophilus} (3); \textit{Lachnoclostridium} (3); \textit{Tyzzerella} (3)
\\
\addlinespace
Healthy & \textit{Blautia} (4); \textit{Erysipelatoclostridium} (4); \textit{Eubacterium} (4); \textit{Parabacteroides} (4); \textit{Roseburia} (4)
\\
\bottomrule 
\end{tabularx}
\end{table}

Next, we apply our method to the vaginal microbiome study \citep{ravel2011vaginal} to understand the association between bacterial vaginosis (BV) and vaginal microbiome compositions. The data, also available at ML Repo \citep{vangayMicrobiomeLearningRepo2019}, comprises $d=241$ microbial taxa counts at the highest available taxonomic resolution across $n=388$ subjects, with 91\% of counts equal to zero. These counts are normalized to compositions. The response variable is the Nugent score (0--10), a Gram stain-based diagnostic index for BV, where scores of 7--10 indicate BV and lower values indicate a healthy vaginal microbiome.

\begin{figure}%[t]
    \centering
    \includegraphics[width=.8\textwidth]{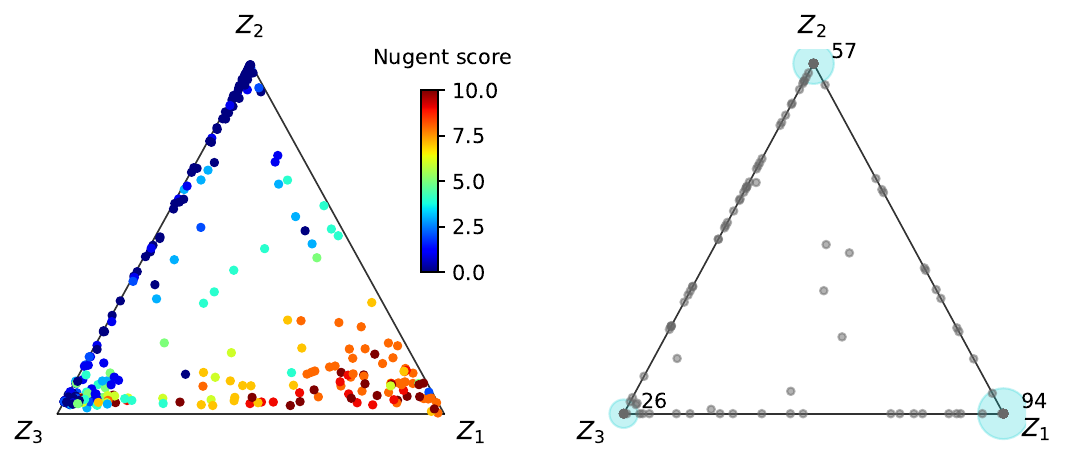}
    \caption{Dual visualization of the vaginal microbiome data, presented similarly to \cref{fig: ileum}.
    }
    \label{fig: ravel}
\end{figure}

\cref{fig: ravel} presents the dual visualization of the vaginal microbiome data. The left panel shows the projection onto $\Delta^2$, which reveals a clear Nugent score gradient: individuals cluster at the left edge with scores of 0, which increase toward the bottom edge as the relative abundance of $z_1$ over $z_2$ grows. Thus, the response is primarily explained by the subcomposition $(z_1, z_2)$, with $z_3$ contributing little. The right panel displays the variable allocation plot, which shows a similar sparsity pattern as before. We interpret the bottom-edge cluster as associated with high Nugent scores and the left-edge cluster as associated with low Nugent scores; representative genera are listed in \cref{tab:vaginal_taxa}. 
These identified taxa align well with prior findings: \textit{Prevotella} and \textit{Peptoniphilus} appear near the bottom edge and are known to be associated with BV, whereas \textit{Lactobacillus} appears near the left edge, consistent with its protective role in maintaining vaginal health \citep{abou2022bacterial}.

\begin{table}%[t]
\centering
\spacingset{1}
\caption{Top 5 frequent genera (with the species counts in parentheses) near the bottom edge (High-NS) and the left edge (Low-NS) from the variable allocation plot of \cref{fig: ravel}, where NS denotes the Nugent score.}
\label{tab:vaginal_taxa}
\footnotesize
\begin{tabularx}{\columnwidth}{l X} % X automatically adjusts the width
\toprule
High-NS & \textit{Anaerococcus} (9); \textit{Corynebacterium} (9); \textit{Prevotella} (9); \textit{Peptoniphilus} (8); \textit{Streptococcus} (8)
\\
\addlinespace
Low-NS & \textit{Streptococcus} (8); \textit{Lactobacillus} (6); \textit{Bacteroides} (5); \textit{Corynebacterium} (5); \textit{Peptoniphilus} (4)
\\
\bottomrule
\end{tabularx}
\end{table}

Finally, we assess the prediction performance of the CKDR method against the existing competitors. Test performance is averaged over 100 random 80/20 train-test splits. For the CD dataset, we measure the misclassification rate (MCR); for the vaginal dataset, we evaluate the mean squared error (MSE). 
\cref{tab: real prediction combined} summarizes the prediction results for both datasets. For the CD classification task, CKDR-5 achieves the lowest average MCR, while CKDR-3 and CKDR$^*$ also significantly outperform LC-Lasso, clr-Kernel, and clr-RF. For the Nugent score regression task, CKDR-5 again yields the lowest average MSE, with CKDR-3 and CKDR$^*$ achieving nearly identical results to CKDR-5. All CKDR variants significantly outperform LC-Lasso and RS-ES, while clr-Kernel and clr-RF show no statistically significant difference from CKDR-5. These consistently superior results confirm that, although CKDR is primarily designed for dimension reduction, its intrinsic predictive model is highly competitive, and in particular, CKDR-3 provides intuitive low-dimensional visualizations that generalize well to unseen data.

\begin{table}[b]
    \spacingset{1}
    \centering
    \caption{Prediction performance for ileum and vaginal microbiome data. The ileum row reports misclassification rate (MCR, \%; standard errors in parentheses), and the vaginal row reports mean squared error (standard errors in parentheses).}
    \label{tab: real prediction combined}
    \scriptsize
    \begin{tabular*}{\columnwidth}{@{\extracolsep{\fill}}lccccccc}
        \toprule
         & CKDR-3 & CKDR-5 & CKDR$^*$ & LC-Lasso & clr-Kernel & clr-RF & RS-ES \\
        \midrule
        Ileum MCR (\%) & 26.3 (0.8) & \textbf{25.5 (0.8)} & 26.3 (0.8) & 28.5 (0.7) & 28.3 (0.6) & 34.5 (0.9) & -- \\
        Vaginal MSE & 3.27 (0.08) & \textbf{3.26 (0.08)} & 3.27 (0.08) & 3.76 (0.06) & 3.50 (0.06) & 3.31 (0.07) & 4.20 (0.09) \\
        \bottomrule
    \end{tabular*}%
\end{table}

\section{Discussions}\label{sec: discussions}

This paper proposes a novel approach for interpretable dimension reduction of compositional data. 
Extending the concept of amalgamation, the CDR framework preserves the simplex geometry, thereby eliminating the need to impute zeros, and its dual visualization effectively conveys an at-a-glance interpretation.
A new composition-tailored SDR is formulated under this framework. For estimation, we develop the CKDR method that yields sparse estimation. We reveal an intrinsic predictive model that enables downstream tasks and optimization via successive convex approximation, and we establish the estimator's consistency through a novel subspace-comparison argument. Applications to microbiome data illustrate the effectiveness of our method in generating predictive low-dimensional visualizations that reveal biologically interpretable patterns. Python codes for the proposed method and experiments are available at \texttt{\url{https://github.com/pjywang/CKDR}}.

\section*{Acknowledgements}
The work of Jeongyoun Ahn and Cheolwoo Park was partially supported by the National Research Foundation of Korea (RS-2022-NR068758).

\putbib
\end{bibunit}

\newpage
\appendix

\renewcommand{\refname}{Supplementary References}
\begin{bibunit}

\renewcommand{\thetable}{S\arabic{table}}  
\renewcommand{\thefigure}{S\arabic{figure}}

\setcounter{figure}{0}
\setcounter{table}{0}
\setcounter{page}{1}

\numberwithin{theorem}{section} \makeatletter \let\c@lemma=\c@theorem \let\c@proposition=\c@theorem \makeatother \renewcommand\thelemma{\thetheorem} 
\renewcommand\theproposition{\thetheorem} 
\renewcommand\thecorollary{\thetheorem}
\renewcommand\thedefinition{\thetheorem}

\vspace*{1em}
\spacingset{1}
\begin{center}
\LARGE{
Supplementary material for ``Geometry-preserving and interpretable dimension reduction of compositional data''
}
\vskip 1em
\end{center}
\bigskip
\begin{abstract}
    \noindent In this supplementary material, we provide technical details of the proposed approach and proofs of the main theoretical results. Section~\ref{sec:supp-sdr} presents technical details and proofs for the proposed compositional SDR framework. Section~\ref{sec:supp-operator-prelim} provides brief preliminaries on random elements in Hilbert spaces. In Section~\ref{sec:supp-kdr}, we develop technical results for the proposed CKDR method introduced in Section~\ref{sec:compositional-kdr}. Finally, the consistency of the CKDR estimator is proved in Section~\ref{sec:supp-consistency-proof}.
    
\end{abstract}
\bigskip

\spacingset{1.5}

\section{Technical details of compositional SDR}\label{sec:supp-sdr}

This section proves the essential results of compositional SDR discussed in \cref{sec:compositional-SDR}. Section~\ref{sec:supp-sdr1} proves the nonexistence of the traditional central subspace with compositional predictors, while Section~\ref{sec:supp-central-compo-subsp} proves the existence of the central compositional subspace. Additionally, Section~\ref{sec:supp-suff-amalg} provides the equivalence between equi-sparse columns in compositional SDR and sufficient amalgamation.
Throughout the section, $X\in\Delta^{d-1}$ is a random compositional predictor variable, $Y$ a random response, and $\supp X$ denotes the support of the distribution $X$ inside $\Delta^{d-1}$. 

\subsection{Nonexistence of the classical central subspace}\label{sec:supp-sdr1}

We prove Lemma~\ref{lemma:zero-central-subs}.
For each $j\in\{1,\ldots, d\}$, define a coordinate-dropping matrix $$B_{-j} = (e_1,\ldots,e_{j-1}, e_{j+1},\ldots,e_d)^\top\in\R^{(d-1)\times d},$$ where the $e_j$ are standard basis vectors in $\R^d$. Note that $B_{-j}$ does not belong to $\cm_{(d-1), d}$ since the $j$th column is zero. These matrices establish the relations
$$Y\ind X\,|\, B_{-j}X\quad \text{for all}\quad j=1,\ldots,d,$$
as the unit-sum constraint on $X \in \Delta^{d-1}$ allows removing each variable $X_j = 1 - \sum_{k \neq j} X_k$ without losing information of $X$. Thus, the matrices $B_{-j}$ are SDR matrices, whose row space is spanned by the vectors $e_1,\ldots,e_{j-1}, e_{j+1},\ldots,e_d$. Therefore, the intersection of all SDR subspaces is always zero, proving that the traditional central subspace does not exist for compositional predictors. \qed

\subsection{Existence of the central compositional subspace}\label{sec:supp-central-compo-subsp}
In this section, we prove that the central compositional subspace $\mathcal{C}_{Y|X}$ exists under a milder assumption (Proposition~\ref{prop:intersection-csdr}). The existence of the central mean compositional subspace $\mathcal{C}_{\E[Y|X]}$ is verified by the same logic, which we omit for brevity.

The argument adapts the classical SDR theory and proceeds similarly, except for the additional existence argument for a CDR matrix $P\in\cm_{m,d}$ corresponding to the intersection of an arbitrary collection of CSDR subspaces---Lemma~\ref{lemma:1subspace} guarantees this, proved at the end of this section. We first introduce a technical but mild condition on subsets of simplices, called $M$-sets, where $M$ stands for ``matching'' \citep{yinSuccessiveDirectionExtraction2008}, adapted to our compositional setting.

\begin{definition}
    A subset $\mathfrak{M}$ of $\Delta^s\times \Delta^t$ is an $M$-set if, for every two pairs $(u, v)$ and $(u', v')$ in $\mathfrak M$, there is a sequence of pairs $(u^{(0)}, v^{(0)}), \ldots, (u^{(l)}, v^{(l)}) $ in $\mathfrak M$ such that (i) $(u^{(0)}, v^{(0)}) = (u, v) $ and $(u^{(l)}, v^{(l)}) = (u', v')$; (ii) for each $i=1, \ldots, l-1$, at least one coordinate remains fixed: $u^{(i)} = u^{(i+1)} $ or $v^{(i)} = v^{(i+1)} $.
\end{definition}

The definition intuitively says that any two pairs $(u,v),\ (u',v')\in\mathfrak{M}$ can be connected by a ``stairway'', where subsequent pairs $(u^{(i)}, v^{(i)})$ and $(u^{(i+1)}, v^{(i+1)})$ share one coordinate value. This is a very mild condition. For example, any open and connected subset $\mathfrak{M}$ of $\Delta^s\times \Delta^t$ is an $M$-set because any two points can be connected by a path, covered by a finite collection of open balls in $\mathfrak{M}$, within which we can locally replace the path with stairways by fixing one component while varying the other. One can also construct disconnected $M$-sets in various scenarios \citep{yinSuccessiveDirectionExtraction2008}.

Returning to CSDR, let $\mathscr{S}_1$ and $\mathscr{S}_2$ be CSDR subspaces of dimensions $m$ and $k$, spanned by rows of $P\in\cm_{m,d}$ and $Q\in\cm_{k,d}$, respectively. Letting $r = \dim(\mathscr{S}_1\cap \mathscr{S}_2)$, we choose a CDR matrix $R\in\cm_{r,d}$ such that $\row(R) = \mathscr{S}_1\cap\mathscr{S}_2$ using Lemma~\ref{lemma:1subspace}. For any point $z$ in the simplex $\Delta^{r-1}$, define
$$\Omega_z = \br*{(Px, Qx)\in\Delta^{m-1}\times\Delta^{k-1}: Rx = z}. $$
Then, the joint distribution $(X,Y)$ is said to satisfy the \textit{$M$-set condition} if $\Omega_z$ is an $M$-set for every projection $z=Rx$ of the point $x\in\Delta^{d-1}$ in the interior of the support of $X$, and for every pair of CSDR subspaces $(\mathscr{S}_1, \mathscr{S}_2)$. 

It is easy to see that under conditions of Proposition~\ref{prop:intersection-csdr}, the joint distribution $(X, Y)$ satisfies the $M$-set condition: letting $S_z = \{x \in \text{rel-int}(\supp X): Rx = z\}$, where rel-int denotes the relative interior to $\Delta^{d-1}$, the slice $S_z$ is path-connected and open relative to the hyperplane $\{x: Rx= z\}$, from which we can construct the stairway in $\Omega_z$ using a finite relative-open cover of a path connecting two points within $S_z$. As mentioned above, the $M$-set condition is much milder than having a convex support with nonempty interior.
Therefore, the following proposition completes the proof of Proposition~\ref{prop:intersection-csdr} under a more general scenario:

\begin{proposition}
    Suppose that the joint pair $(X, Y)$ satisfies the $M$-set condition. Then, the intersection of any collection of CSDR subspaces is itself a CSDR subspace.
\end{proposition}

The proof of this proposition essentially parallels Proposition 6.4 of \citet{cook1998regression}, as also noted in \citet{yinSuccessiveDirectionExtraction2008}, and is therefore omitted. The only substantive difference arises from the geometry \textit{relative to} the simplex, while the classical SDR relies on the geometry in the ambient Euclidean space. In particular, the $M$-set argument for classical SDR fails for compositional predictors because openness relative to the simplex does not translate to openness in the Euclidean setting. This dimension deficiency violates the Euclidean version of the $M$-set condition for compositions and enables the counterexample in Lemma~\ref{lemma:zero-central-subs}.
\qed

\subsubsection{Proof of Lemma~\ref{lemma:1subspace}}
    We prove the existence by explicit construction of a CDR matrix $P\in\cm_{m,d}$ with $\row(P) = V$. Note that this existence is automatic in classical Euclidean SDR, but it becomes nontrivial in our nonnegative, unit-sum-constrained framework, necessitating some geometric arguments.
    
    Let $W_c$ denote the affine hyperplane $\{x\in\R^d\,|\,x_1+\cdots + x_d = c\} $ in $\R^d$ for $c\in\R$. Denote 
    $$V' := V \cap W_0 = V\cap W_d - \1_d$$ by the $m-1$ dimensional subspace of $V$ without the vector $\1_d$.

    Pick any basis vectors $u_1,\ldots, u_{m-1}$ that spans $V'$, and let $u_m = -(u_1+\cdots+u_{m-1})$. Then, choose a sufficiently large number $N > 0 $ so that every vector 
    $$v_i:= u_i + N \1_d \in V$$
    has strictly positive components. Then, the $v_i$ are \textit{linearly independent}, as the vectors $\{u_i\}_{i=1}^m$ are affinely independent and $\1_d$ is not contained in their linear span. The $v_i$ are thus positive vectors spanning the subspace $V$ since $m=\dim V$. As we have the equality
    $v_1+\cdots+v_m = mN\1_d,$
    the matrix $$P = (\frac{1}{mN}v_1,\ldots \frac{1}{mN}v_m)^\top$$ is a column-stochastic matrix contained in $\cm_{m,d}$. This CDR matrix $P$ has positive entries and $\row(P) = V$, completing the proof. \qed

\subsection{Equi-sparsity and sufficient amalgamation}\label{sec:supp-suff-amalg}
This section provides an equivalence between the equi-sparsity structure of columns in compositional SDR and sufficient amalgamation mentioned in Remark~\ref{rmk:suff-amalgam}.

Sufficient amalgamation is defined by $Y \ind X \,|\,AX$, where $A\in\cm_{m,d}$ is a binary CDR matrix, thus a binary CSDR matrix. Let $\mathcal{A}_{Y|X}$ denote the \emph{central amalgamation subspace}, defined as the intersection of the row spaces of all binary CSDR matrices. This minimal subspace partitions the variables of $X$, corresponding to a partition of the index set $[d]=\{1,\ldots, d\}$. The following lemma establishes a natural connection between equi-sparsity in CSDR and sufficient amalgamation:

\begin{lemma}\label{lemma:suff-amalgam}
    Suppose the rows of $P\in\cm_{m,d}$ span the central compositional subspace $\mathcal{C}_{Y|X}$. Define the partition $\mathcal{P}(P)$ of $[d]$ by grouping indices according to identical columns: $\mathcal{P}(P) = \{I\subseteq[d]: P_i = P_j \text{ for all }(i, j)\in I\times I\}$. Then $\mathcal{P}(P) = \mathcal{A}_{Y|X}$.
\end{lemma}

This result is similar to the sufficient variable selection in the sparse SDR literature \citep{yinSequentialSufficientDimension2015, zengSubspaceEstimationAutomatic2024}, where sparsity enables recovery of the minimal sufficient set of predictors. Analogously, in the compositional setting, equi-sparsity in CSDR leads to sufficient amalgamation, identifying groups of functionally similar variables that can be merged without information loss.

\subsubsection{Proof of Lemma~\ref{lemma:suff-amalgam}}

For any partitions $\mathcal{P}_1$ and $\mathcal{P}_2$ of $[d]$, denote $\mathcal{P}_1\le \mathcal{P}_2$ if $\mathcal{P}_1$ is coarser than $\mathcal{P}_2$, which defines a partial order of partitions. We will prove the two inequalities $\mathcal{P}(P)\le \mathcal{A}_{Y|X}$ and $\mathcal{P}(P) \ge \mathcal{A}_{Y|X}$.

To begin, let $e_1,\ldots, e_d \in \R^d$ be the standard basis vectors of $\R^d$. For each subset of indices $I\subseteq [d]$, define $e_I = \sum_{i:i\in I} e_i$, a binary vector with 1's at the indices of $I$. Writing $\mathcal{P}(P) = \{I_1,\ldots, I_s\}$, we can form a binary CDR matrix 
$$A = [e_{I_1},\ldots, e_{I_s}]^\top\in\cm_{s, d}.$$
By construction, the amalgamation matrix $A$ has the same column equality structures as $P$, forming the same partition $\mathcal{P}(A) = \mathcal{P}(P)$ from the columns. Thus, the rows of $P$ are linear combinations of the $e_{I_j}$; i.e., $\row(P)\subseteq \row(A)$. Since $P$ is a matrix satisfying the SDR relation, $A$ also satisfies $Y\ind X\,|\,AX$ by Proposition 2.3 of \citet{li2018sufficient}, establishing the inclusion of sufficient amalgamation subspaces $$\mathcal{A}_{Y|X}\subseteq \row (A).$$
At the corresponding partition level of amalgamation subspaces, this inclusion indicates that the partition $\mathcal{A}_{Y|X}$ is coarser than $\row(A)$, and thus $$\mathcal{A}_{Y|X}\le \mathcal{P}(A) = \mathcal{P}(P).$$

For the reverse inequality, let $A'$ be another binary CDR matrix with $\row(A') = \mathcal{A}_{Y|X}$. Letting $\mathcal{P}(A') = \{J_1,\ldots, J_t\}$, which is coarser than $\mathcal{P}(P)$, we can similarly assume that $A' = [e_{J_1},\ldots, e_{J_t}]^\top\in\cm_{t, d}$. Then, since $\row(P) \subseteq \row(A')$ due to the minimality of $\row(P)$, the rows of $P$ are linear combinations of the binary vectors $e_{J_k}$. Thus, if $i, j\in J_k$ for some $J_k$, then $P_i = P_j$ holds. This essentially shows that each $J_k$ is contained in one of the index sets $I_l$ of $\mathcal{P}(P)$, proving the reverse inequality 
$$\mathcal{P}(P) \le \mathcal{P}(A') = \mathcal{A}_{Y|X}.$$ This completes the proof of the equality $\mathcal{P}(P) = \mathcal{A}_{Y|X}$. 
\qed

\section{Preliminaries on random elements in a Hilbert space}\label{sec:supp-operator-prelim}

Before pursuing the technical exposition of our CKDR method, we introduce some preliminary notions regarding random elements in separable Hilbert spaces, along with their mean and covariance. For further properties and proofs related to these notions, see \citet{hsingTheoreticalFoundationsFunctional2015}.

Given a probability space $(\Omega,\P)$ with a Borel $\sigma$-field and a real separable Hilbert space $(\mathcal{H},\langle\cdot,\cdot\rangle_{\mathcal{H}})$, a measurable mapping \(F\colon \Omega \longrightarrow \mathcal{H}\)
is called a \emph{random element} on $\ch$. In our RKHS $(\ch_\cx, k_\cx)$ on $\cx$ and a random vector $X\in\cx$, the RKHS embedding $\Phi:= k_\cx(X,\cdot)\in\ch_\cx$ defines a random element on $\ch_\cx$, which is of our central interest. 
Since $\E[k_\cx(X, X)] < \infty$, we always have the finite second moment: $\E[\|\Phi\|_{\ch_\cx}^2] < \infty$, where $\|\cdot\|_\ch$ denotes the norm on $\ch$.

If \(\E[\|F\|_{\mathcal{H}}] < \infty,\) , the random element $F$ is \textit{Bochner integrable} (Section 2.6 of \citet{hsingTheoreticalFoundationsFunctional2015}), defining a \textit{mean element} of $F$ via:
$$\E[F] := \int_\Omega F d\P.$$
The mean element is characterized by its inner products:
\[
\langle \E[F],\, h \rangle_{\mathcal{H}}
=
\E[\langle F,\,h\rangle_{\mathcal{H}}],
\quad
\text{for all }h \in \mathcal{H},
\]
where the equality follows from the fact that the Bochner integral is interchangeable with bounded linear functionals.

Consider another Hilbert space $(\cg,\langle\cdot,\cdot\rangle_\cg)$ and a random element $G\in\cg$. If two second moments are bounded, $\E[\|F\|_{\mathcal{H}}^2], \E[\|G\|_\cg^2] < \infty$, the \emph{cross-covariance operator} is defined by
\[
\s_{GF} := \E[(G - \E[G])\otimes (F - \E[F])]\in \cg\otimes\ch,
\]
where any rank-one operator $y\otimes x\in\cg\otimes\ch$ acts as
\[
(y\otimes x)\,h
=
\langle x,h\rangle_{\mathcal{H}}\,y,
\quad h\in\mathcal{H}.
\]
The tensor product space $\cg\otimes\ch$ is isometric to the space of Hilbert-Schmidt operators, hence the norm is given by:
$$\|\s_{GF}\|_{HS}^2 = \|\E[(G - \E[G])\otimes (F - \E[F])]\|_{\cg\otimes\ch}.$$
In case $\ch=\cg$ and $F= G$, $\s_{FF}$ is called the \textit{covariance operator}, which is self-adjoint, positive semi-definite, and trace-class with 
\[
\Tr(\s_{FF})
=
\E[\|F - \E[F]\|_{\mathcal{H}}^2]
<\infty.
\]

Finally, we note that by plugging in $G=k_\cy(Y,\cdot)\in\ch_\cy$ and $F=k_\cx(X,\cdot)\in\ch_\cx$, the operator $\s_{GF}$ coincides with the cross-covariance operator $\s_{YX}$ defined on RKHSs in \eqref{eq: cov-def}.

\section{Compositional KDR formulation (\cref{sec:compositional-kdr})}\label{sec:supp-kdr}

This section provides technical details about our CKDR method given in \cref{sec:compositional-kdr}. 
In Section~\ref{sec:supp-kdr-sdrproof} we provide the proof of Proposition~\ref{thm:sdr-population-guarantee}. Section~\ref{sec:supp-kdr-pullback} provides the compatibility result between our target-domain formulation and the prior KDR development (Remark~\ref{rmk:re-embedding}). In Section~\ref{sec:supp-vv-kdr-proofs} we prove the results on the intrinsic predictive model of KDR, introduced in \cref{sec: vv-krr}. Note that the results in this section can naturally extend to general Euclidean settings, including the classical Stiefel manifold and beyond. 

We begin with a key equality regarding the conditional covariance operator $\s_{YY|X}$, which holds whenever the kernel $k_\cx$ is characteristic \citep{fukumizuKernelDimensionReduction2009}:
\begin{equation}
    \langle g, \s_{YY|X} g \rangle_{\ch_\cy} = \E[\var[g(Y)|X]]. \label{eq:condvar}
\end{equation}
This equality is instrumental in proving the results of this section and Section~\ref{sec:supp-consistency-proof}.

\subsection{Proof of Proposition~\ref{thm:sdr-population-guarantee}}\label{sec:supp-kdr-sdrproof}

The equation \cref{eq:residual} and the $L^2$-density of $\ch_\cx$ implies that
    \begin{align*}
        \langle g, \s_{YY|p(X)}g \rangle_{\ch_\cy} &= \inf_{h\in\ch_\cz} \var(g(Y) - h(p(X)))  \\
        &\overset{(*)}{=} \inf_{f\in\ch_\cx^p}\var(g(Y) - f(X))\quad\ \\
        &\ge \inf_{f\in\ch_\cx}\var(g(Y) - f(X)) = \langle g, \s_{YY|X}g\rangle_{\ch_\cy},
    \end{align*}
    which proves the L\"owner ordering of operators. In the second line, the space $\ch_\cx^p$ denotes the RKHS associated with the pullback kernel $k_\cx^p(x, x') := k_\cz(p(x), p(x'))$. The second equality $(*)$ holds since $\ch_\cx^p=\{h\circ p:\cx\to\cz \,|\, h\in\ch_\cz \}$ (see pullback theorem in Section~\ref{sec:supp-kdr-pullback}), and the inequality holds since $\ch_\cx^p$ is continuously embedded in $L^2(\P_X)$ due to the boundedness assumption $\E[k_\cz(Z, Z)] < \infty$ (\cref{sec:conditional-covariance}).
    
    For the equality case, we use \eqref{eq:condvar} under the characteristicity assumption of $k_\cz$:
    \begin{equation}\label{eq:diag-condcov-comparison}
        \langle g, (\s_{YY|p(X)} - \s_{YY|X})g \rangle_{\ch_\cy} = \E[\var(g(Y)|p(X))] - \E[\var(g(Y)|X)]. 
    \end{equation}
    Letting $Z = p(X)$, the law of total variance implies that
    $$\var(g(Y)|Z) = \E[\var(g(Y)|X, Z)|Z] + \var(\E[g(Y)|X, Z]|Z),$$
    which yields
    \begin{align*}
        \E[\var(g(Y)|Z)] &= \E[\E[\var(g(Y)|X, Z)|Z]] + \E[\var(\E[g(Y)|X, Z]|Z)]\\
        &=\E[\var(g(Y)|X, Z)] + \E[\var(\E[g(Y)|X, Z]|Z)] \\ 
        &= \E[\var(g(Y)|X)] + \E[\var(\E[g(Y)|X]|Z)],
    \end{align*}
    where the last equality uses the inclusion of the $\sigma$-fields $\sigma(Z)\subset\sigma(X)$. Then, we have
    \begin{align*}
        \s_{YY|Z} = \s_{YY|X} &\iff \E[\var(\E[g(Y)|X]|Z)] = 0,\ \forall g\in\ch_\cy \\
        &\iff \var(\E[g(Y)|X]|Z) = 0 \quad\text{a.s. } \forall g\in\ch_\cy\\
        &\iff \E[g(Y)|X] = \E[g(Y)|Z] \quad\text{a.s. } \forall g\in\ch_\cy. 
    \end{align*}
    Based on this equivalence, we prove parts (i) and (ii) of the SDR guarantees.
    
\paragraph*{Part (i): SDR guarantee.}
    The assumption that $k_\cy$ is characteristic ensures that for all measurable set $A\subset \cy$, the indicator function $\chi_A$ on $A$ is approximated by $\ch_\cy$-functions up to a constant; i.e., $$\E[g(Y)|X] = \E[g(Y)|Z] \text{ a.s., } \forall g\in\ch_\cy \iff \P_{Y|X} = \P_{Y|Z}, $$
    where the last equality is equivalent to the SDR $Y\ind X \,|\, Z$. \qed

\paragraph*{Part (ii): SDR for conditional mean.}
    In this case, note that our finiteness assumption for the kernels implicitly assumes that $\E\|Y\|^2_\ch < \E[\langle Y, Y\rangle_\ch] <\infty$, assuring the existence of the vector-valued mean of $Y$ via Jensen's inequality. Also, the linear kernel enables identifying $\ch_\cy$ as a subspace of $\ch$; we thus write $g(Y) = \langle g, Y\rangle_\ch = \langle g, Y\rangle_{\ch_\cy}$ for all $g\in\ch_\cy$ by abusing notations.
    
    Suppose first that $\s_{YY|Z} = \s_{YY|X}$. Since any continuous linear functional commutes with Bochner integration (see e.g., Theorem 3.1.7 of \citet{hsingTheoreticalFoundationsFunctional2015}), the following equality holds almost surely: for all $g\in\ch_\cy$,
    $$\langle g, \E[Y|X]\rangle_{\ch_\cy} = \E[g(Y)|X],$$
    which implies
    $$\langle g, \E[Y|X]\rangle_{\ch_\cy} = \langle g, \E[Y|Z]\rangle_{\ch_\cy}.$$
    Considering a CONS $g_1, g_2, \ldots$ of $\ch_\cy$, we can construct an almost-sure region on which the equality $\langle g, \E[Y|X]\rangle_{\ch_\cy} = \langle g, \E[Y|Z]\rangle_{\ch_\cy}$ holds for all $g\in\ch_\cy$, by linearity and continuity. Therefore,
    $$\E[Y|X] = \E[Y|Z] \quad\text{a.s. }\ \text{on }\ \ch_\cy\subseteq\ch.$$

    Conversely, if $\E[Y|X] = \E[Y|Z] \text{ almost surely}$, we can follow the previous proof in the reverse direction, yielding the equality
    $$\langle g, \E[Y|X]\rangle_\ch = \E[g(Y)|X]\quad\text{a.s.}$$
    for all $g\in\ch_\cy$. This proves the equality $\s_{YY|Z} = \s_{YY|X}$, completing the proof.
    \qed

\begin{remark}
    In this result, the role of the kernel $k_\cx$ on the original domain $\cx$ is only to provide a lower bound for the conditional covariance operator after projection, $\s_{YY|p(X)}$. The requirement that $\ch_\cx$ is dense in $L^2(\P_X)$ can be satisfied by many kernels, including $L^2$-universal kernels; see \citet{sriperumbudur2011universality} for details.
\end{remark}

\begin{remark}
    The original result of \citet{fukumizuKernelDimensionReduction2009} assumes that $\ch_\cx^p + \R$ is dense in a certain $L^2$-space on $\cx$, whereas we simply assume that $k_\cz$ is characteristic. 
\end{remark}

\subsection{Compatibility of KDR formulations (Remark~\ref{rmk:re-embedding})}\label{sec:supp-kdr-pullback}

In this section, we demonstrate the compatibility between two KDR formulations: the classical KDR using the RKHS $\ch_\cx^p$, and our target-based approach that uses the RKHS $\ch_\cz$. 
As noted in Remark~\ref{rmk:re-embedding}, these different RKHSs give rise to two conditional covariance operators $\s_{YY|p(X)}$ and $\s_{YY|X}^p$, with their associated cross-covariance and correlation operators.
The compatibility result of this section thus enables us to adopt many established results from \citet{fukumizuKernelDimensionReduction2009} to our target-based setting, while avoiding the re-embedding assumption that limits the generalization of KDR beyond Stiefel manifolds.

Here, we interpret the RKHS $\ch_\cx^p$ as a \textit{pullback} of $\ch_\cz$ and establish that the cross-covariance and the correlation operators can likewise be ``pullbacked.'' Crucially, we prove the equality between the conditional covariance operators $\s_{YY|p(X)}$ and $\s_{YY|X}^p$ and the same equality for their empirical counterparts. Although intuitive, the rigorous account of the compatibility requires understanding the interplay between the covariance operators and the \textit{pullback operator} arising from the projection map $p:\cx\to\cz$.

We first state the \textit{pullback theorem} \cite[Theorem~2.9]{saitoh2016theory} which describes the exact members of the RKHS $\ch_\cx^p$ associated with the kernel $k_\cx^p(x, x'):= k_\cz(p(x), p(x'))$:
\begin{equation}\label{eq: pullback RKHS}
    \ch_\cx^p = \{f:\cx\to\R\,|\, f = g\circ p \text{ for some } g\in\ch_\cz \}.
\end{equation}
The equality \eqref{eq: pullback RKHS} defines a \textit{pullback operator} $p^*:\ch_\cz\to\ch_\cx^p$, sending $g\in\ch_\cz$ to $g\circ p \in\ch_\cx^p$. By the pullback theorem, the pullback operator is bounded and surjective, indicating that essential information in $\ch_\cx^p$ can be completely recovered in the target RKHS $\ch_\cz$. However, the map $p^*$ is not necessarily injective, especially when $p$ is not surjective; for instance, if $g\in\ch_\cz$ is supported outside the image of $p$, it becomes zero inside $\ch_\cx^p$.
We will see, nonetheless, that at the covariance operators level, the essential covariance information is not lost, thereby establishing the equivalence at the conditional covariance level.

We start with the compatibility result related to the cross-covariance operators $\s_{Yp(X)}$ and $\s_{YX}^p$ defined on the different domains, $\ch_\cx^p$ and $\ch_\cz$. The following lemma establishes the coherence between these operators:

\begin{lemma}\label{lemma: pullback coherent}
    Write $Z=p(X)$. The pullback operator $p^*:\ch_\cz\to\ch_\cx^p$ and the covariance operators are coherent, making the following diagrams commutative:
    \begin{equation*}    
        \begin{tikzcd}[column sep=large, row sep=large] % enlarge the row arrow length
            \ch_\cx^p \arrow{dr}{\s_{YX}^p} &  & \ch_\cx^p & & \ch_\cx^p \arrow{r}{\s_{XX}^p} & \ch_\cx^p \\            
            \ch_\cz \arrow{r}[swap]{\s_{YZ}}\arrow{u}{p^*} & \ch_\cy  & \ch_\cz \arrow{u}{p^*} & \ch_\cy\arrow{l}{\s_{ZY}}\arrow{ul}[swap]{\s_{XY}^p} & \ch_\cz\arrow{r}[swap]{\s_{ZZ}}\arrow{u}{p^*} & \ch_\cz.\arrow{u}{p^*}
        \end{tikzcd}
    \end{equation*}
\end{lemma}

\begin{proof}
    Let $\phi(Z) = k_\cz(Z,\cdot)\in\ch_\cz$ and $\psi(Y) = k_\cy(Y,\cdot)\in\ch_\cy$ be the embedded random elements in the RKHSs. Write their means as
    $$m_Z = \E[\phi(Z)] \quad\text{and}\quad m_Y =\E[\psi(Y)], $$
    also known as kernel mean embeddings \citep{muandet2017kernel} of distributions $Z$ and $Y$. Denoting the centered elements by $\wt\phi(Z) = \phi(Z) - m_Z$ and $\wt\psi(Y) = \psi(Y) - m_Y$, the cross-covariance operator $\s_{YZ}$ can be written as:
    $$\s_{YZ} = \E[\wt\psi(Y)\otimes\wt\phi(Z)] \in \ch_\cy\otimes\ch_\cz.$$
    Since $Z=p(X)$, we can explicitly pullback the $\phi(Z)$ as
    $$p^*\phi(Z) = k_\cx(p(X), p(\cdot)) = k_\cx^p(X, \cdot)\in\ch_\cx^p, $$
    and $m_Z$ is pullbacked similarly as $$p^*m_Z = p^*\E[\phi(Z)] = \E[p^*\phi(Z)] = \E[k_\cx^p(X, \cdot)]\in\ch_\cx^p. $$
    Here, the commutativity between $p^*$ and the expectation $\E$ holds because $p^*$ is bounded.
    The pullback $p^*\wt\phi(Z)$ thus coincides with the centered kernel mean embedding of $X$ via $k_\cx^p$:
    $$p^*\wt\phi(Z) = k_\cx^p(X, \cdot) - \E[k_\cx^p(X,\cdot)]\in\ch_\cx^p.$$
    Therefore, 
    $$\s_{YX}^p = \E[\wt\psi(Y)\otimes p^*\wt\phi(Z)]\in\ch_\cy\otimes \ch_\cx^p.,$$
    which establishes the equality $\s_{YX}^p p^* = \s_{YZ}.$ The remaining results are proved similarly. 
\end{proof}

By plugging in the empirical distributions to this result, the coherence result also holds for the empirical operators $\hatsig_{Z*}$ and their pullbacks $\hatsig_{X*}^p$. From this, we can prove the equality between the regularized empirical conditional covariance operators on $\ch_\cy$: since
\begin{align*}
    \hatsig_{YX}^p (\hatsig_{XX}^p + \varepsilon_nI)^{-1}\hatsig_{XY}^p &=
    \hatsig_{YX}^p (\hatsig_{XX}^p + \varepsilon_nI)^{-1} p^*\hatsig_{ZY} \\
    &= \hatsig_{YX}^p p^*(\hatsig_{ZZ}+\varepsilon_nI)^{-1}\hatsig_{ZY} \\
    &= \hatsig_{YZ} (\hatsig_{ZZ}+\varepsilon_nI)^{-1}\hatsig_{ZY},
\end{align*}
we have the equality of empirical conditional covariance operators:
$$\hatsig_{YY|p(X)} = \hatsig_{YY|X}^p.$$
This equality is used in \cref{sec:ckdr-method} where we adopted the same computation strategy as in the classical KDR methods.

Next, we state an analogous result for the correlation operators $V_{YZ}$ and $V_{YX}^p$. The following lemma is established using the uniqueness property of such operators given in \eqref{eq: correlation}.

\begin{lemma}\label{lemma: pullback corr}
    Write $Z=p(X)$. Let $V_{YZ}$ and $V_{YX}^p$ be the correlation operators satisfying
    $$\s_{YZ} = \s_{YY}^{1/2}V_{YZ}\s_{ZZ}^{1/2}\quad\text{and}\quad \s_{YX}^p = \s_{YY}^{1/2}V_{YX}^p(\s_{XX}^p)^{1/2}. $$
    Then, the correlation operators are coherent with the pullback operator $p^*$:
    $$V_{YZ} = V_{YX}^p p^*\quad\text{and}\quad V_{XY}^p = p^*V_{ZY}.$$
\end{lemma}

\begin{proof}
     We first prove the similar commutativity to Lemma~\ref{lemma: pullback coherent} for the square-root operators $\s_{ZZ}^{1/2}$ and $(\s_{XX}^p)^{1/2}$. 
    Consider the spectral decomposition of $\s_{ZZ}$ with a CONS $\{e_i\}\subset\clo\ran(\s_{ZZ})$:
    $$\s_{ZZ} = \sum_{i=1}^\infty \lambda_i e_i\otimes e_i \in \ch_\cz\otimes \ch_\cz.$$
    Using Lemma~\ref{lemma: pullback coherent}, we have:
    \begin{equation}\label{eq: s_XX^p}
        \s_{XX}^p = \sum_{i=1}^\infty \lambda_i p^*(e_i)\otimes p^*(e_i) \in\ch_\cx^p \otimes \ch_\cx^p. 
    \end{equation}
    Note that we have the inclusion $\clo\ran(\s_{ZZ})\subseteq (\ker p^*)^\perp$ since, for all $h\in\ch_\cz$ and $l\in\ker p^*$, we have
    $$\langle l, \s_{ZZ}h\rangle_{\ch_\cz} = \cov[h(Z), l(p(X))] = 0. $$
    Thus, from the isometry $(\ker p^*)^\perp\cong\ch_\cx^p$ along the pullback operator $p^*$, the inner products are preserved:
    $$\langle p^*(e_i),p^*(e_j)\rangle_{\ch_\cx^p} = \delta_{ij},$$ meaning that the equality \eqref{eq: s_XX^p} is also a spectral decomposition of $\s_{XX}^p$. Then, using the same eigenfunctions, we have similar representations of the square-root operators $\s_{ZZ}^{1/2}$ and $(\s_{XX}^p)^{1/2}$. This implies that: $$p^*\s_{ZZ}^{1/2} = (\s_{XX}^p)^{1/2}p^*,$$
    paralleling the commutativity at the covariance operator level.
    
    Building on this equality, we draw the following diagram:
    \begin{equation*}
        \begin{tikzcd}[column sep=6em, row sep=1.5em]
            \ch_\cx^p\arrow{r}{(\s_{XX}^p)^{1/2}} & \clo\ran(\s_{XX}^p)\arrow{rd}{\!\!\!\!V_{YX}^p} & & \\
            & & \clo\ran(\s_{YY})\arrow{r}{\s_{YY}^{1/2}} & \ch_\cy, \\
            \ch_\cz \arrow{uu}{p^*} \arrow{r}{\s_{ZZ}^{1/2}} & \clo\ran(\s_{ZZ}) \arrow{uu}{p^*}\arrow{ur}[swap]{\!\!\!V_{YZ}} & &
        \end{tikzcd}
    \end{equation*}
    where the square on the left-hand side is commutative. From our results, all the paths from $\ch_\cz$ to $\ch_\cy$ in this diagram are equal to the operator $\s_{YZ} = \s_{YZ}^p p^*$. In particular, we have
    $$\s_{YY}^{1/2}V_{YX}^pp^*\s_{ZZ}^{1/2} = \s_{YY}^{1/2}V_{YZ}\s_{ZZ}^{1/2}.$$
    Then, by the \textit{uniqueness property} of the operator $V_{YZ}$ given in \eqref{eq: correlation}, we must have the equality $V_{YZ} =V_{YX}^p p^* $, as desired. The other equality is proved symmetrically.
\end{proof}

Finally, we establish the equivalence of the conditional covariance operators using the commutativity results at the covariance and the correlation operators:

\begin{lemma}\label{lemma:condcov-equality}
        The two conditional covariance operators on $\ch_\cy$ coincide:
        \begin{equation*}%\label{eq: equality condcov}
        \s_{YY|p(X)} = \s_{YY|X}^p.
        \end{equation*}
\end{lemma}

\begin{proof}
    Lemma~\ref{lemma: pullback corr} implies that
    $$V_{YZ}V_{ZY} = V_{YX}^p p^*V_{ZY} = V_{YX}^p V_{XY}^p, $$
    which completes the proof of the equality $$\s_{YY|p(X)} = \s_{YY|X}^p.$$
\end{proof}

\subsection{Results on the intrinsic predictive model of CKDR}\label{sec:supp-vv-kdr-proofs}

This section gives the details for the discussions in \cref{sec: vv-krr}. We prove the equivalence between the CKDR empirical objective and the vector-valued KRR with \textit{intercept} in \eqref{eq:empirical-kdr=vv-krr}. To our knowledge, there is no formal study on vector-valued KRR with an intercept. We thus provide a rigorous proof regarding the existence and uniqueness of the solution of such a regression problem, which in turn facilitates the proof of the equivalence result (Theorem~\ref{thm:trace=centered-vv-krr}) in Section~\ref{sec:supp-vv-krr-equivalence}. Below, recall that $H = I_n - \frac{1}{n}\1\1^T$ is the centering matrix.

\begin{proposition}\label{prop: vv-ckrr}
    Let $(z_1, \psi_1),\ldots, (z_n, \psi_n)\in\cz\times\ch_\cy$ be given data, and let $\Psi = (\psi_1,\ldots,\psi_n)^\top\in(\ch_\cy)^n$ denote the column vector. Let $\mathcal G_\cz$ be an $\ch_\cy$-valued RKHS induced by $k_\cz$. If $F^\star\in\mathcal{G}_\cz$ and $\gamma^\star \in\ch_\cy$ minimizes the loss function
    $$L_n(F, \gamma) = \frac{1}{n}\sum_{i=1}^n \norm*{\psi_i - F(z_i) - \gamma}_{\ch_\cy}^2 + \varepsilon_n\norm*{F}_{\mathcal{G}}^2,$$
    then such minimizers are unique, and $F^\star$ has the form 
    \begin{equation}\label{eq: vv-ckrr-solution}
        F^\star(\cdot) = \sum_{i=1}^n k_\cz(z_i, \cdot) \alpha_i :\cz\to\ch_\cy,
    \end{equation}
    where $\alpha = (\alpha_1,\ldots,\alpha_n)^\top = (G_Z+n\varepsilon_nI_n)^{-1}H \Psi \in(\ch_\cy)^n $, and $G_Z$ is the centered Gram matrix formed by $z_1,\ldots, z_n$. Also, the intercept is determined as $\gamma^\star = \frac{1}{n}\sum_{i=1}^n(\psi_i - F^\star(z_i))$.
\end{proposition}

\subsubsection{Proof of Proposition~\ref{prop: vv-ckrr}}\label{sec: proof of prop6}

\begin{proof}
    Observe first that for any fixed $F \in \mathcal G$, $L_n(F,\gamma)$ is minimized by the mean value $ \gamma = \frac{1}{n}\sum_{i=1}^n(\psi_i - F(z_i)) \in\ch_\cy$. By letting $v_i = \psi_i - F(z_i)$ and $\overline{v} = \frac{1}{n}\sum_i v_i$,
    \begin{align*}
        L_n(F,\gamma) &= \frac{1}{n}\sum_{i=1}^n\|v_i - \gamma\|_{\ch_\cy}^2 + \varepsilon_n\|F\|_\mathcal{G}^2\\
        &= \frac{1}{n}\sum_{i=1}^n \|v_i - \overline{v}\|_{\ch_\cy}^2 + \|\overline{v} - \gamma\|_{\ch_\cy}^2 + \varepsilon_n\|F\|_\mathcal{G}^2.
    \end{align*}
    Thus, the unique value $\gamma$ that minimizes $L_n(F,\gamma)$ for a fixed $F$ is $\overline{v} = \frac{1}{n}\sum_{i=1}^n(\psi_i - F(z_i))$. We set $R_n(F) = L_n(F, \frac{1}{n}\sum_{i=1}^n(\psi_i - F(z_i)))$ the loss function depending only on $F\in\mathcal G_\cz$.
    
    Let $F^\star$ be as defined in \eqref{eq: vv-ckrr-solution}, and let $\eta = F - F^\star$ be an $\ch_\cy$-valued function on $\cz$ for an arbitrary $F\in\mathcal G_\cz$. We then compute the loss function as:
    \begin{align*}
        R_n(F) &= R_n(\eta + F^\star)\\
        &= \frac{1}{n}\sum_{i=1}^n \norm*{\psi_i - (\eta + F^\star)(z_i) - \frac{1}{n}\sum_{j=1}^n(\psi_j - (\eta + F^\star)(z_j))}_{\ch_\cy}^2 + \varepsilon_n\|\eta + F^\star\|_{\mathcal G_\cz}^2\\
        &= R_n(F^\star) + \frac{1}{n}\sum_{i=1}^n\left\|\eta(z_i) - \frac{1}{n}\sum_{j=1}^n\eta(z_j)\right\|_{\ch_\cy}^2 \\
        & \qquad - \frac{2}{n}\sum_{i=1}^n\left\langle \psi_i - F^\star(z_i) - \frac{1}{n}\sum_{j=1}^n(\psi_j - F^\star(z_j)),\ \eta(z_i) - \frac{1}{n}\sum_{j=1}^n\eta(z_j)\right\rangle_{\ch_\cy} \\ 
        & \qquad + 2\varepsilon_n \langle \eta, F^\star\rangle_{\mathcal G_\cz} + \varepsilon_n \|\eta\|_{\mathcal G_\cz}^2.
    \end{align*}
    As $F^\star(\cdot) = \sum_{i=1}^n k_\cz(z_i, \cdot) \alpha_i$, where $\alpha = (\alpha_1,\ldots,\alpha_n)^\top = (G_Z+n\varepsilon_nI_n)^{-1}H \Psi$, we can compute the inner products using the reproducing property. The latter inner product is readily computed as
    $$\langle \eta, F^\star\rangle_{\mathcal G_\cz} = \langle \eta, \sum_{i=1}^n k_\cz(z_i, \cdot)\alpha_i\rangle_{\mathcal{G}_\cz} = \sum_{i=1}^n\langle \eta(z_i), \alpha_i\rangle_{\ch_\cy}. $$
    To compute the other inner product, observe first that 
    \begin{equation}\label{eq: commutative H}
        (G_Z+n\varepsilon_nI_n)^{-1} H = H(G_Z + n\varepsilon_nI_n)^{-1},
    \end{equation}
    which implies $H\alpha =\alpha $, meaning that $\alpha_1+\cdots +\alpha_n =0$. Then, 
    \begin{align*}
        F^\star(z_i) - \frac{1}{n}\sum_{j=1}^n F^\star(z_j) &= \sum_{l=1}^n \pr*{k_\cz(z_i, z_l) - \frac{1}{n}\sum_{j=1}^nk_\cz(z_j,z_l) }\alpha_l \\
        &= e_i^\top HK_Z \alpha\\
        &= e_i^\top G_Z \alpha
    \end{align*}
    since $H\alpha = \alpha$ and $G_Z = HK_ZH$. Using the relation $\psi_i - \frac{1}{n}\sum_{j=1}^n \psi_j = e_i^\top (G_Z+n\varepsilon_nI_n)\alpha$ from the definition of $\alpha$, we have 
    \begin{align*}
        \psi_i - F^\star(z_i) - \frac{1}{n}\sum_{j=1}^n(\psi_j - F^\star(z_j)) &= e_i^\top (G_Z+n\varepsilon_nI_n)\alpha - e_i^\top G_Z\alpha \\
        &=n\varepsilon_n\alpha_i,
    \end{align*}
    and thus, the inner product becomes
    \begin{align*}
        &\frac{2}{n}\sum_{i=1}^n\left\langle \psi_i - F^\star(z_i) - \frac{1}{n}\sum_{j=1}^n(\psi_j - F^\star(z_j)),\ \eta(z_i) - \frac{1}{n}\sum_{j=1}^n\eta(z_j)\right\rangle_{\ch_\cy} \\
        & \qquad = \frac{2}{n} \sum_{i=1}^n\left\langle n\varepsilon_n\alpha_i,\eta(z_i) - \frac{1}{n}\sum_{j=1}^n\eta(z_j)\right\rangle_{\ch_\cy}\\
        & \qquad = 2\varepsilon_n \sum_{i=1}^n \left\langle \alpha_i, \eta(z_i) - \frac{1}{n}\sum_{j=1}^n\eta(z_j)\right\rangle_{\ch_\cy} \\
        & \qquad = 2\varepsilon_n\sum_{i=1}^n \langle \eta(z_i), \alpha_i \rangle_{\ch_\cy},
    \end{align*}
    where the last equality is derived from the fact that $\sum_i \alpha_i = 0$. Therefore,
    \begin{align*}
        R_n(F) = R_n(F^\star) + \frac{1}{n}\sum_{i=1}^n\left\|\eta(z_i) - \frac{1}{n}\sum_{j=1}^n\eta(z_j)\right\|_{\ch_\cy}^2 + \varepsilon_n \|\eta\|_{\mathcal G_\cz}^2,
    \end{align*}
    which is minimized if and only if $\eta = F - F^\star = 0$.
\end{proof}

\subsubsection{Proof of Theorem~\ref{thm:trace=centered-vv-krr}}\label{sec:supp-vv-krr-equivalence}
\begin{proof}
    The equivalence is established by directly computing the minimized loss function
    $$\text{(RHS)}\quad\underset{F\in\mathcal{G}_\cz,\,
        \gamma\in\ch_\cy}{\min}\ \frac{1}{n}\sum_{i=1}^n \norm*{k_\cy(y_i, \cdot) - F(p(x_i)) - \gamma}_{\ch_\cy}^2 + \varepsilon_n\norm*{F}_{\mathcal{G}_\cz}^2,$$
    where RHS stands for the right-hand side of the equality of \eqref{eq:empirical-kdr=vv-krr}. Setting $\psi_i = k_\cy(y_i,\cdot)$ and $z_i = p(x_i)$, Proposition~\ref{prop: vv-ckrr} gives the unique minimizers $(F_P, \gamma_P)$ with $\alpha_i$ as defined there, so that $\psi_i - F_P(z_i) - \frac{1}{n}\sum_{j=1}^n(\psi_j - F_P(z_j)) = n\varepsilon_n\alpha_i$. Then,
    \begin{align*}
        \text{(RHS)} &= n\varepsilon_n^2\sum_{i=1}^n\|\alpha_i\|_{\ch_\cy}^2 + \varepsilon_n \norm*{\sum_{i=1}^n k_\cz(z_i,\cdot)\alpha_i}_{\cg_\cz}^2\\
        &= \varepsilon_n\pr*{n\varepsilon_n \Psi^\top H(G_Z+n\varepsilon_nI_n)^{-2}H\Psi + \Psi^\top H(G_Z+n\varepsilon_nI_n)^{-1}K_Z(G_Z+n\varepsilon_nI_n)^{-1}H\Psi}\\
        &\overset{(*)}{=} \varepsilon_n \Tr\pr*{G_Y(G_Z+n\varepsilon_nI_n)^{-1}(n\varepsilon_n(G_Z+n\varepsilon_nI_n)^{-1} + G_Z(G_Z+n\varepsilon_nI_n)^{-1}}\\
        &= \varepsilon_n \Tr(G_Y(G_Z+n\varepsilon_nI_n)^{-1}),
    \end{align*}
    where the equality with (*) indicates the equality \eqref{eq: commutative H} is used. Since $\Tr(\hatsig_{YY|p(X)}) = \varepsilon_n \Tr(G_Y(G_Z+n\varepsilon_nI_n)^{-1}),$ the proof of the equation \eqref{eq:empirical-kdr=vv-krr} is complete, and thus so is \eqref{eq:joint-learning}.
\end{proof}

\subsection{Convergence of the SCA algorithm}\label{sec:supp-sca-convergence}

This section proves Proposition~\ref{prop:sca-convergence} under the theoretical condition $\epsilon_{\mathrm{tol}} = 0$. Note that while general SCA convergence relies on the strong convexity of the surrogate to uniformly bound the directions \citep[e.g.,][]{scutariParallelDistributedMethods2017}, the compactness of $\cm_{m,d}$ enables this guarantee utilizing only the continuous differentiability of our convex surrogate. 

The surrogate satisfies the following properties established in \cref{sec:optimization-sca}:
\begin{enumerate}
    \item[(S1)] $Q_t(P)$ is convex and continuously differentiable in $P$;
    \item[(S2)] $Q_t(P)$ is continuous jointly in $(P, P_t)$;
    \item[(S3)] $\nabla Q_t(P_t) = \nabla T_n(P_t)$ (gradient consistency).
\end{enumerate}
Since the regularization $\varepsilon_n > 0$ ensures that $G_{PX} + n\varepsilon_n I_n$ is uniformly positive definite on the compact set $\cm_{m,d}$, the objective $T_n$ is continuously differentiable with Lipschitz-continuous gradient on $\cm_{m,d}$.

\begin{proof}
Let $\wt{P}_t \in \argmin_{P \in \cm_{m,d}} Q_t(P)$ and $D_t = \wt{P}_t - P_t$.

\noindent\textbf{Step 1: Stationarity or descent direction.} 
If $Q_t(\wt{P}_t) = Q_t(P_t)$, then $P_t$ minimizes the convex function $Q_t$ over $\cm_{m,d}$. Consequently, $\langle \nabla Q_t(P_t), P - P_t \rangle \ge 0$ for all $P \in \cm_{m,d}$.
By gradient consistency~(S3), this is precisely the first-order stationarity condition for $T_n$ at $P_t$. In this case, the algorithm terminates.

If $Q_t(\wt{P}_t) < Q_t(P_t)$, convexity of $Q_t$ yields
\[
\langle \nabla T_n(P_t), D_t \rangle
= \langle \nabla Q_t(P_t), D_t \rangle
\le Q_t(\wt{P}_t) - Q_t(P_t) < 0,
\]
confirming that $D_t$ is a descent direction for the true objective $T_n$.

\noindent\textbf{Step 2: Well-defined Armijo line search.}
Since $T_n$ is continuously differentiable and $D_t$ is a descent direction, Taylor expansion gives $T_n(P_t + \lambda D_t) = T_n(P_t) + \lambda \langle \nabla T_n(P_t), D_t \rangle + o(\lambda)$. Because $c \in (0, 1)$, the Armijo inequality 
\[
T_n(P_t + \lambda D_t) \le T_n(P_t) + c\lambda \langle \nabla T_n(P_t), D_t\rangle
\]
holds for all sufficiently small $\lambda > 0$, guaranteeing the backtracking search terminates after finitely many halvings.

\noindent\textbf{Step 3: Objective convergence.}
The accepted step ensures strict decrease: $T_n(P_{t+1}) \le T_n(P_t) + c\lambda_t \langle \nabla T_n(P_t), D_t \rangle < T_n(P_t)$. Since the continuous function $T_n$ is bounded below on the compact domain $\cm_{m,d}$, the sequence $\{T_n(P_t)\}$ converges, and the successive decrements vanish: $T_n(P_t) - T_n(P_{t+1}) \to 0$.

\noindent\textbf{Step 4: Accumulation points are stationary.}
Suppose, for the sake of contradiction, that a subsequence $P_{t_k}\to\bar P$, where $\bar P$ is not first-order stationary. Then there exists $P'\in\cm_{m,d}$ such that
\[
\langle \nabla T_n(\bar P), P' - \bar P\rangle < 0.
\]
Write the surrogate as $Q(P;Y)$ when its base point is $Y$, so that $Q_t(P)=Q(P;P_t)$. By gradient consistency, $\nabla_P Q(\bar P;\bar P)=\nabla T_n(\bar P)$. Since $Q(\cdot;\bar P)$ is continuously differentiable, there is an $\alpha\in(0,1]$ such that $P_\alpha:=\bar P+\alpha(P'-\bar P)\in\cm_{m,d}$ and
\[
Q(P_\alpha;\bar P) < Q(\bar P;\bar P).
\]
Hence, for some $\eta>0$ and all sufficiently large $k$, the joint continuity of $Q$ gives
\[
Q(P_\alpha;P_{t_k}) \le Q(P_{t_k};P_{t_k})-\eta .
\]
Since $\wt{P}_{t_k}$ minimizes $Q(\cdot;P_{t_k})$ over $\cm_{m,d}$,
\[
Q(\wt{P}_{t_k};P_{t_k}) - Q(P_{t_k};P_{t_k}) \le -\eta.
\]
Convexity of $Q(\cdot;P_{t_k})$ then yields
\[
\langle \nabla T_n(P_{t_k}), D_{t_k}\rangle
= \langle \nabla_P Q(P_{t_k};P_{t_k}), D_{t_k}\rangle
\le Q(\wt{P}_{t_k};P_{t_k}) - Q(P_{t_k};P_{t_k})
\le -\eta .
\]
Let $L$ be a Lipschitz constant of $\nabla T_n$ on $\cm_{m,d}$ and let $R=\operatorname{diam}(\cm_{m,d})$. The $L$-Lipschitz continuity of the gradient guarantees that the Armijo condition holds whenever
\[
\lambda \le \frac{2(1-c)\eta}{LR^2}.
\]
Thus the backtracking rule accepts step sizes bounded below by $\lambda_{\min}:=\min\{1,\beta\,2(1-c)\eta/(LR^2)\}>0$ along the subsequence. Consequently, the Armijo condition gives
\[
T_n(P_{t_k}) - T_n(P_{t_k+1}) \ge c\,\lambda_{\min} \eta > 0
\]
for all sufficiently large $k$. This contradicts $T_n(P_t) - T_n(P_{t+1}) \to 0$ from Step 3. Therefore, every accumulation point must be a first-order stationary point. 
\end{proof}

\section{Consistency of the CKDR estimator (\cref{sec:consistency})}\label{sec:supp-consistency-proof}

This section proves our consistency result in \cref{sec:consistency} over the set of CDR matrices $\cm_{m,d}$. In Section~\ref{sec:supp-counterexample}, we also give an illustration of why the population objective is discontinuous, precluding the classical uniform convergence argument of KDR.

We first introduce some notations for convenience:
$$T_n(P):= \Tr(\hatsig_{YY|PX}) \quad\text{and}\quad T(P):=\Tr(\s_{YY|PX}), $$
where $T_n(P)$ is computed as in the equation \eqref{eq:empirical-objective}. Since the empirical objective $T_n$ is regularized by the parameter $\varepsilon_n$ and $T$ is not regularized, we also introduce an intermediate bridge, the regularized function at the population level: for $\varepsilon >0$, 
$$T^{\varepsilon}(P):= \Tr\pr*{\s_{YY} - \s_{Y,PX}(\s_{PX,PX}+\varepsilon I)^{-1}\s_{PX,Y} },$$
where $I$ denotes the identity operator. Using these notations, we prove our consistency result based on the following three key results:
\begin{enumerate}
    \item[(i)] In Section~\ref{sec:supp-part-i}, we prove the uniform convergence between $T_n$ and $T^{\varepsilon_n}$ (Corollary~\ref{cor: part (i)}):
    $$\sup_{P\in\cm_{m,d}}\, \abs*{T_n(P) - T^{\varepsilon_n}(P)} = O_p\pr*{\frac{1}{\varepsilon_n\sqrt{n}}}.$$
    \item[(ii)] In Section~\ref{sec:supp-part-ii}, we show that $T^{\varepsilon_n}(P)$ \textit{monotonically} converges to $T(P)$ (Lemma~\ref{cco:lemma-pointwise convergence}), and that $T$ is continuous on each rank-$k$ subset $\cm_{m,d}^{(k)}$ of $\cm_{m,d}$ (Lemma~\ref{lemma:continuity-population}).
    \item[(iii)] In Section~\ref{sec:supp-part-iii}, we complete the consistency proof by establishing a pointwise convergence $T(\wh P_n) \to T(P^\star)$ first, followed by showing that the minimum of $T$ is \textit{well-separated} from ``bad regions''
    (Lemma~\ref{lemma:uniform-gaps}).
\end{enumerate}

The uniform rate part (i) largely follows the corresponding result of \citet{fukumizuKernelDimensionReduction2009}, based on the compatibility results in Section~\ref{sec:supp-kdr-pullback}. 
Part (ii) extends similar results in prior work, while we derive \textit{monotonic pointwise convergence} and \textit{rank-wise continuity} rather than uniform convergence between $T^\varepsilon$ and $T$, which is invalid in our varying-rank domain $\cm_{m,d}$. The monotonicity is crucial in deriving the convergence $T(\wh P_n)\to T(P^\star)$ (Corollary~\ref{cor:key-convergence}), and the rank-wise continuity is also essential in proving a uniform-gap result in part (iii). We finish the consistency proof in part (iii), which establishes the pointwise convergence and uniform gap results outlined in Remark~\ref{rmk:consistency-proof}.

\subsection{Part (i): uniform rate with the intermediate function}\label{sec:supp-part-i}

In the following lemma, $\|\Sigma\|_{HS}$ denotes the Hilbert-Schmidt norm of the operator $\Sigma$ on a Hilbert space, and $\|\Sigma\|$ denotes the operator norm of $\Sigma$. Its proof is given in the reference:

\begin{lemma}[{\citet[Lemma 8]{fukumizuKernelDimensionReduction2009}}]\label{cco:lemma-decomposition}
    For $P\in\cm_{m,d}$, write $Z=PX$. Then,
    \begin{align*}
        & \abs*{T_n(P) - T^{\varepsilon_n}(P)} \\[0.5mm]
        & \quad \leq \frac{1}{\varepsilon_n}\br*{(\|{\hatsig_{YZ}}\|_{HS} + \|\s_{YZ}\|_{HS}) \|\hatsig_{YZ} - \s_{YZ}\|_{HS} + \Tr(\s_{YY})\|\hatsig_{ZZ} - \s_{ZZ}\|} \\
        & \qquad\quad + \abs*{\Tr(\hatsig_{YY} - \s_{YY})}.
    \end{align*}
\end{lemma}

Since the operator norm is bounded by the Hilbert-Schmidt norm (spectral theorem), $\|\hatsig_{ZZ} - \s_{ZZ}\|\leq \|\hatsig_{ZZ}-\s_{ZZ}\|_{HS}$, the following lemma suffices to establish part (i).

\begin{lemma}\label{cco:lemma-sup convergence}
    Under Assumption~\ref{assumption: Lipschitz}, all the terms
    $$\sup_{P\in\cm_{m,d}} \|\hatsig_{Y,PX} - \s_{Y,PX}\|_{HS},\ \sup_{P\in\cm_{m,d}} \|\hatsig_{PX,PX} - \s_{PX,PX}\|_{HS},\ \text{ and }\ \abs*{\Tr(\hatsig_{YY} - \s_{YY})} $$
    are of order $O_p\big(\frac{1}{\sqrt{n}}\big)$ as $n\to\infty$.
\end{lemma}

\begin{proof}
    For each $P\in\cm_{m,d}$, set $p(x)=Px$ and use the pullback RKHS $\ch_\cx^p$ induced by $k_\cx^p(x,x')=k_\cz(p(x),p(x'))$.
    Assumption~\ref{assumption: Lipschitz} is precisely condition (A-3) of \citet{fukumizuKernelDimensionReduction2009}, and Lemma~9 therein directly gives
    $$\sup_{P\in\cm_{m,d}}\|\hatsig_{YX}^{p}-\s_{YX}^{p}\|_{HS}=O_p(n^{-1/2}),\qquad
    \sup_{P\in\cm_{m,d}}\|\hatsig_{XX}^{p}-\s_{XX}^{p}\|_{HS}=O_p(n^{-1/2}).$$
    Since $p^*$ induces an isometry on $(\ker p^*)^\perp$, by results in Section~\ref{sec:supp-kdr-pullback}, the pullback via $p^*$ preserves the corresponding Hilbert-Schmidt norms. Therefore, these two HS norms are exactly the desired HS norms for $\hatsig_{Y,PX}-\s_{Y,PX}$ and $\hatsig_{PX,PX}-\s_{PX,PX}$. By Lemma~9 of \citet{fukumizuKernelDimensionReduction2009}, we have all three terms $O_p(n^{-1/2})$.
\end{proof}

Lemmas \ref{cco:lemma-decomposition} and \ref{cco:lemma-sup convergence} complete the proof of the following uniform rate between $T_n$ and the intermediate bridge function $T^{\varepsilon_n}$:

\begin{corollary}\label{cor: part (i)}
    Under Assumption~\ref{assumption: Lipschitz} and the condition \eqref{eq: epsilon condition} on the regularization parameter $\varepsilon_n$, we have the uniform rate
    $$\sup_{P\in\cm_{m,d}}\, \abs*{T_n(P) - T^{\varepsilon_n}(P)} = O_p\pr*{\frac{1}{\varepsilon_n\sqrt{n}}}, $$
    as $n\to\infty$.
\end{corollary}

\subsection{Part (ii): properties of $T^{\varepsilon_n}$ and $T$}\label{sec:supp-part-ii}

Next, we study the properties between the bridge function $T^{\varepsilon}$ and the population objective $T$. Intuitively, $T^{\varepsilon}$ behaves like a smoothed version of $T$: as the ridge parameter $\varepsilon$ shrinks, the smoothing vanishes and the function $T^{\varepsilon}$ approaches $T$.
The next lemma shows the monotonic pointwise convergence of $T^\varepsilon\to T$, which is essential in our consistency proof:
\begin{lemma}\label{cco:lemma-pointwise convergence}
    Whenever $\varepsilon > \varepsilon' > 0$, we have $T^\varepsilon \ge T^{\varepsilon'}$. Moreover, for each $P\in\cm_{m,d}$, 
    $T^{\varepsilon}(P)\to T(P) $ as $\varepsilon\to 0$.
\end{lemma}
\begin{proof}
    Let $P\in\cm_{m,d}$, and set $Z = PX$. Recall that 
    $$T^\varepsilon(P) = \Tr(\s_{YY} - \s_{YZ}(\s_{ZZ} + \varepsilon I)^{-1}\s_{ZY}),$$
    so we can write the difference as
    \begin{align*}
        T^\varepsilon(P) - T^{\varepsilon'}(P) &= \Tr\bigl(\s_{YZ}\{(\s_{ZZ} + \varepsilon' I)^{-1} - (\s_{ZZ} + \varepsilon I)^{-1}\}\s_{ZY}\bigr) \\
        &= (\varepsilon - \varepsilon')\Tr(\s_{YZ}(\s_{ZZ} + \varepsilon' I)^{-1}(\s_{ZZ} + \varepsilon I)^{-1}\s_{ZY}),
    \end{align*}
    which is nonnegative due to the positivity of the operators $(\s_{ZZ} + \varepsilon I)^{-1}$ and $(\s_{ZZ} + \varepsilon' I)^{-1}$. 
    
    The proof of pointwise convergence $T^{\varepsilon}(P)\to T(P)$ can be directly adopted from Lemma 11 of \citet{fukumizuKernelDimensionReduction2009}.
\end{proof}

The next lemma establishes the continuity of the population objective function $T(P)$ on \textit{each} rank-$k$ subset $\cm_{m,d}^{(k)}$ of $\cm_{m,d}$. %Our focus on the target RKHS $\ch_\cz$ simplifies the corresponding proof in the classical KDR methods. 

\begin{lemma}\label{lemma:continuity-population}
    Suppose that Assumption~\ref{assumption:modified-continuity} holds and that $k_\cz$ is characteristic. Then, $T(P)$ is continuous on each $\cm_{m,d}^{(k)}$, $k=1,\ldots, m$.
\end{lemma}

\begin{proof}
    By taking a CONS of $\ch_\cy$ and applying the dominant convergence theorem, it suffices to show that the mapping $P\mapsto \langle g, \s_{YY|PX}\,g\rangle_{\ch_\cy}$ is continuous on $\cm_{m,d}^{(k)}$ for any $g\in\ch_\cy$. Since $k_\cz$ is characteristic, we apply the equality \eqref{eq:condvar}, which yields:
    \begin{align*}
        \langle g, \s_{YY|PX}g\rangle_{\ch_\cy} &= \E[\var[g(Y)|PX]]\\
        &=\E[g(Y)^2] - \E[\E[g(Y)|PX]^2].
    \end{align*}
    Thus, the desired continuity is equivalent to the continuity of $P\mapsto \E[\E[g(Y)|PX]^2] = \E[\E[g(Y)|\Pi_{\row(P)}X]^2] $ on each $\cm_{m,d}^{(k)}$. Since the set of continuous bounded functions on $\cy$ is dense in $L^2(\P_Y)$, which contains the RKHS $\ch_\cy$, we may assume that $g$ is continuous and bounded on $\cy$.

    For such $g$, Assumption~\ref{assumption:modified-continuity} implies the continuity of the mapping $V\mapsto \E[\E[g(Y)|\Pi_{V}X]^2] $ on each $\Gr^\1(k, d)$. To extend this continuity to the matrix level, we only need to check the continuity of the mapping
    $$\cm_{m,d}^{(k)}\to\Gr^\1(k,d);\quad P\mapsto \row(P).$$ 
    Here, the row space projection map $\Pi_{\row(P)}$ is represented by the matrix $P^\top(PP^\top)^\dagger P$, where $\dagger$ indicates the Moore-Penrose pseudoinverse. As the association $A\mapsto A^\dagger$ is continuous on any set of matrices of \textit{fixed rank}, it completes the proof.
\end{proof}

We emphasize that the pointwise convergence $T^\varepsilon\to T$, the monotonicity $T^\varepsilon \ge T$, and the rank-wise continuity of $T$ will suffice for our consistency proof, unlike the uniform convergence-based proofs in the related KDR literature \citep{fukumizuKernelDimensionReduction2009, chenTheoryFeatureLearning2026}.

\subsection{Consistency proof: pointwise convergence and uniform gaps}\label{sec:supp-part-iii}

We now prove our main consistency result. Throughout this section, we pick a minimizer $P^\star$ of the population function $T(P)$ on $\cm_{m,d}$, whose existence is guaranteed by the existence of the central compositional subspace $\mathcal{C}_{Y|X}$ and Proposition~\ref{thm:sdr-population-guarantee} (with Assumption~\ref{assumption:kernels-and-dimension}). We assume that $P^\star$ satisfies $\row(P^\star) = \mathcal{C}_{Y|X}$ using Lemma~\ref{lemma:1subspace}, and write the global minimum as:
$$T_0 := T(P^\star) = \min_{P\in\cm_{m,d}}T(P).$$
As mentioned in Remark~\ref{rmk:consistency-proof}, we prove Theorem~\ref{thm: consistency} by establishing 
\begin{enumerate}
    \item the pointwise convergence $T(\wh P_n)\to T_0$ (Corollary~\ref{cor:key-convergence}); and
    \item the \textit{uniform separation} of the minimum $T_0$ from ``two bad regions'' (Lemma~\ref{lemma:uniform-gaps}).
\end{enumerate}
The pointwise convergence is derived from the results in Sections \ref{sec:supp-part-i} and \ref{sec:supp-part-ii}, while the uniform separation Lemma~\ref{lemma:uniform-gaps} requires a novel Grassmannian projection argument, whose proof is deferred to the end of this section.

\subsubsection{Pointwise convergence}

We begin with establishing the following one-sided convergence corollary, which is in turn equivalent to the pointwise convergence $T(\wh P_n)\to T_0$ due to the minimality of $T_0$:
\begin{corollary}\label{cor:key-convergence}
    Suppose the same assumptions of Theorem~\ref{thm: consistency}. For any positive number $\eta > 0$, we have
    $$\P(T(\wh P_n) \le T_0 + \eta)\to 1 \quad\text{as}\quad n\to\infty.$$
\end{corollary}

\begin{proof}
    As $T_n(\wh P_n)\le T_n(P^\star)$ and $T_n(P^\star)\to T(P^\star)=T_0$ in probability (Corollary~\ref{cor: part (i)} and Lemma~\ref{cco:lemma-pointwise convergence}), we have
    \begin{equation*}
        T_n(\wh P_n)\le T_0 + o_P(1).
    \end{equation*}
    By the \textit{uniform control} between $T_n$ and $T^{\varepsilon_n}$ (Corollary~\ref{cor: part (i)}), we get $|T_n(\wh P_n ) - T^{\varepsilon_n}(\wh P_n)|\to 0$ in probability, implying that
    \begin{equation*}
        T^{\varepsilon_n}(\wh P_n) \le T_n(\wh P_n) + o_P(1).
    \end{equation*}
    Combining these two inequalities and the \textit{monotonicity} $T \le T^{\varepsilon_n}$ (Lemma~\ref{cco:lemma-pointwise convergence}) yields
    \begin{equation}\label{eq:key-asymp-ineq}
        T(\wh P_n) \le T_0 + o_P(1),
    \end{equation}
    which deduces the desired one-sided convergence
    $$\P(T(\wh P_n) \le T_0 + \eta)\to 1.$$
\end{proof}

\subsubsection{Uniform separation results and the proof of Theorem~\ref{thm: consistency}}

To prove Theorem~\ref{thm: consistency}, we introduce two ``bad regions'' that the estimator must avoid: the low-rank subset $\cm_{m,d}^{(<m^\star)}:=\{P\in\cm_{m,d}:\rk(P) < m^\star\}$ and the distant subset $K_\delta:= \{P\in\cm_{m,d}: \rho(\row(P), \mathcal{C}_{Y|X})\ge \delta\}$, where $m^\star = \dim \mathcal{C}_{Y|X}$.
We establish the following key uniform separation lemma, whose proof is deferred to Section~\ref{sec:supp-separation-proof}.

\begin{lemma}\label{lemma:uniform-gaps}
    Under Assumptions~\ref{assumption:kernels-and-dimension} and \ref{assumption:modified-continuity}, we have the strict inequalities:
    \begin{equation}\label{eq:uniform-gaps}
        \inf_{P\in\cm_{m,d}^{(< m^\star)}} T(P) > T_0 \quad\text{and}\quad \inf_{P\in K_\delta} T(P) > T_0.
    \end{equation}
\end{lemma}

With this lemma, the proof of Theorem~\ref{thm: consistency} follows immediately.

\paragraph*{Proof of Theorem~\ref{thm: consistency}.}
By Lemma~\ref{lemma:uniform-gaps}, we set $$\eta = \min\bigg(\inf_{P\in\cm_{m,d}^{(< m^\star)}} T(P) - T_0,\ \inf_{P\in K_\delta} T(P) - T_0\bigg) > 0.$$ 
On the event that either $\rk(\wh P_n) < m^\star$ or $\rho(\row(\wh P_n), \mathcal{C}_{Y|X}) \ge \delta$ occurs, we must have $T(\wh P_n) \ge T_0 + \eta$.
However, by Corollary~\ref{cor:key-convergence}, the probability of the disjoint event $T(\wh P_n) \le T_0 + \eta / 2$ converges to 1.
Therefore, the probability that $\wh P_n$ falls into either bad region converges to 0, which yields:
$$\lim_{n\to\infty}\P\left(\rk(\wh P_n)\ge m^\star\ \wedge\ \rho(\row(\wh P_n), \mathcal{C}_{Y|X}) < \delta\right) = 1,$$
completing the proof. \qed

\subsubsection{Proof of Lemma~\ref{lemma:uniform-gaps}}\label{sec:supp-separation-proof} 

We now prove the key Lemma~\ref{lemma:uniform-gaps}. A direct analysis of the infima over $\cm_{m,d}$ is hindered by the discontinuity of $T(P)$ at rank-deficient limits and the non-compactness of the subsets. We circumvent this issue by projecting these sets into the union of the Grassmannians $\Gr^\1(k, d)$, $k=1, \ldots, m$.

For each $k=1,\ldots, m$, consider the row space mapping from the rank-$k$ subset $\cm_{m,d}^{(k)}$ to the compact manifold $\Gr^\1(k,d)$:
$$\Pi_k:\cm_{m,d}^{(k)} \to \Gr^\1(k,d);\quad P\mapsto \row(P), 
$$
which is \textit{surjective} by Lemma~\ref{lemma:1subspace}. Denoting $\cs$ by the disjoint union $\bigcup_{k=1}^m\Gr^\1(k,d)$, the set of subspaces of $\R^d$ containing $\1_d$, there is a natural extension of $\Pi_k$:
$$\Pi:\cm_{m,d}\to \cs; \quad P\mapsto \row(P),
$$
which is again \textit{surjective}. We identify $\Pi(P) = \Pi_{\row(P)}$, the orthogonal projection matrix onto $\row(P)$.

Given a CONS $g_1,g_2,\ldots$ of $\ch_\cy$, we formally define a function $J:\cs\to\R$ by:
$$J(V):= \sum_{i\ge 1} \E[\var(g_i(Y)|\Pi_VX)]. $$
Because $k_\cz$ is characteristic, the equality \eqref{eq:condvar} yields $T(P) = J(\Pi(P))$ on $\cm_{m,d}$. By surjectivity of $\Pi$, the minimizer $P^\star$ of $T$ attains the minimum of $J$ on $\cs$; i.e., 
$$\min_{V\in\cs} J(V) = J(\Pi(P^\star)) = T_0.$$
Furthermore, by Assumption~\ref{assumption:modified-continuity}, each summand $\E[\var(g_i(Y)|\Pi_VX)]$ of $J$ is continuous on $\Gr^\1(k, d)$. As we have assumed $\E[k_\cy(Y,Y)] < \infty$ throughout the paper, the sum of these summands is bounded by a finite number: $\sum_{i\ge 1} \E[g_i(Y)^2] = \E[\|k_\cy(Y, \cdot)\|^2_{\ch_\cy}] = \E[k_\cy(Y,Y)] < \infty$. Thus, $J$ is \textit{continuous} on each $\Gr^\1(k, d)$.

We first prove $\inf_{P\in\cm_{m,d}^{(< m^\star)}} T(P) > T_0$. Note that the image of $\cm_{m,d}^{(<m^\star)}$ under $\Pi$ is exactly $\bigcup_{k=1}^{m^\star - 1} \Gr^\1(k, d)$. Thus:
$$\inf_{P\in\cm_{m,d}^{(< m^\star)}} T(P) = \min_{1 \le k < m^\star} \inf_{V\in\Gr^\1(k, d)} J(V).$$
Because $J$ is continuous on each compact manifold $\Gr^\1(k, d)$, the infimum on the right-hand side is attained at some subspace $V^*$ with $\dim(V^*) < m^\star$. If this infimum equals $T_0$, i.e., $J(V^*) = T_0$, any matrix $P_V$ with $\Pi(P_V) = V^*$ satisfies $T(P_V) = T_0$. By Proposition~\ref{thm:sdr-population-guarantee}, such a matrix $P_V$ satisfies compositional SDR, making $V^*$ a CSDR subspace. Hence, by the definition of the central compositional subspace, $\mathcal{C}_{Y|X}\subseteq V^*$. This contradicts $\dim(V^*) < m^\star = \dim(\mathcal{C}_{Y|X})$. Thus, the strict inequality holds.

Next, we prove $\inf_{P\in K_\delta} T(P) > T_0$. The set $K_\delta$ projects via $\Pi$ to 
$$F_\delta = \{V\in\cs: \rho(V, \mathcal{C}_{Y|X})\ge \delta \}.$$
On each $\Gr^\1(k, d)$, the distance $\rho(V, \mathcal{C}_{Y|X})$ is continuous in $V$, making the intersection $F_\delta \cap \Gr^\1(k, d)$ a closed subset of a compact space, and hence compact \citep{yeSchubertVarietiesDistances2016}.
As $J$ is continuous, it attains its minimum on $F_\delta$, denoted by $J(V_\delta)$ for some $V_\delta\in F_\delta$. Since $V_\delta \not\supset \mathcal{C}_{Y|X}$ by the distance condition $\rho(V_\delta, \mathcal{C}_{Y|X}) \ge \delta$, we obtain the strict inequality $J(V_\delta) > T_0$: if we had $J(V_\delta) = T_0$, then any CDR matrix $P_\delta$ with $\Pi(P_\delta) = V_\delta$ would satisfy compositional SDR by Proposition~\ref{thm:sdr-population-guarantee}, leading to a contradictory inclusion $V_\delta \supseteq \mathcal{C}_{Y|X}$.
Since $K_\delta = \Pi^{-1}(F_\delta)$, we conclude $$\inf_{P\in K_\delta} T(P) = \inf_{V\in F_\delta} J(V) = J(V_\delta) > T_0.$$ \qed

\subsection{Counterexample to uniform convergence over the rank-variable CDR domain}\label{sec:supp-counterexample}

In this section, we illustrate why the population objective $T(P)$ on $\cm_{m,d}$ essentially has discontinuities, occurring when a sequence of matrices converges to a lower-rank matrix. This discontinuity not only invalidates the classical uniform convergence argument but also indicates that Assumption (A-1) of \citet{fukumizuKernelDimensionReduction2009} cannot directly apply to our compositional domain $\cm_{m,d}$, which led to our modified subspace dimension-wise Assumption~\ref{assumption:modified-continuity}.

For ease of illustration, we set $Y$ to be a univariate variable in $\R$, endowed with the linear kernel $k_\cy(y, y') = yy'$. Then, for any $P\in\cm_{m,d}$, we have
\begin{align*}
    \Tr(\s_{YY|PX}) &= \E[\var(Y|PX)] \\
    &= \E[Y^2] - \E[\E[Y|PX]^2],
\end{align*}
whenever $k_\cz$ is characteristic by the equality \eqref{eq:condvar}. The continuity of $T(P)$ is thus equivalent to the continuity of the mapping $P\mapsto \E[\E[Y|PX]^2]$ on $\cm_{m,d}$; note that this equivalence can extend to general response kernels $k_\cy$.

Then, we give a concrete example of discontinuity. Let $U\sim\mathcal{U}(0, 1)$ be a uniform random variable, let $X = (U, 1- U)^\top\in\Delta^1$, and let $Y = U$. We design the following rank-degenerating sequence of CDR matrices that map $\Delta^1\to\Delta^1$:
$$ P = \frac{1}{2}\begin{pmatrix}
    1 & 1 \\ 1 & 1
\end{pmatrix},\quad P_n = P + \frac{1}{n}\begin{pmatrix}
    1 & -1 \\ -1 & 1
\end{pmatrix},$$
which are CDR matrices satisfying $\rk(P_n) = 2$ for all $n\ge 2$. Then, since $PX$ is always the constant vector $\1_2 /2$, we have
\begin{align*}
    \E[Y|P_nX] & = \E[Y|X] = \E[Y|U] = U \\
    \E[Y|PX] &= \E[Y] = \frac{1}{2},
\end{align*}
which gives
$$\E[\E[Y|PX]^2] = \frac{1}{4},\quad\text{and}\quad \E[\E[Y|P_nX]^2] = \frac{1}{3}.$$
Therefore, $T(P_n)$ does not converge to $T(P)$, and thus $T$ is not continuous.

One can obtain countless such discontinuity examples by creating a sequence of CDR matrices that converges to a lower-rank matrix. Intuitively, the rank drop causes an abrupt reduction of the residual information in $Y$ after being described by $P_nX$, caused by a reduction of independent directions over which $P_nX$ can vary, resulting in discontinuities as above. Such concrete counterexamples confirm that the prior KDR theory based on uniform convergence works only on the fixed-rank Stiefel manifold.

\begin{remark}\label{rmk:counterexample-previous}
    The same example also provides a counterexample in the weak-limit step used in a related rank-collapse prevention argument for KDR \citep[Lemma~12]{chenTheoryFeatureLearning2026}. 
    The proof obtains a weak $L^2$ accumulation point of $\{\E[Y| P_nX]\}_{n=1}^\infty$ and then implicitly treats it as $\sigma(PX)$-measurable, thereby identifying it with $\E[Y| PX]$. In our example, however, $\E[Y| P_nX]=U$ for every $n$, so any weak $L^2$ limit is $U$, whereas $PX$ is constant and $U$ is not $\sigma(PX)$-measurable. Our Grassmannian-projection argument in Lemma~\ref{lemma:uniform-gaps} avoids this issue while establishing the rank-collapse prevention rigorously.
\end{remark}

\putbib
\end{bibunit}
\end{document}